\documentclass[a4paper]{amsproc}
\usepackage{amssymb}
\usepackage{amscd}
\usepackage[dvips]{graphicx}
\usepackage[dvips]{graphicx}

\newtheorem{theorem}{Theorem}[section]

\newtheorem{lemma}[theorem]{Lemma}

\theoremstyle{definition}
\newtheorem{definition}[theorem]{Definition}
\newtheorem{remark}{Remark}

\newtheorem{exm}{Example}

\def\ji {\char'032}
\def\ja {\char'037}
\def\m  {\char'176}

\font\rrm=wncyr8%
\font\rit=wncyi8

\newcommand{\A}{A^{-1}}
\newcommand{\R}{\mathbb{R}}
\newcommand{\N}{\mathbb{N}}

\newcommand{\p}{\partial}
\newcommand{\ds}{\displaystyle}

\DeclareMathOperator{\Span}{\mathrm{span\,}}
\DeclareMathOperator{\diag}{\mathrm{diag}}

\title[VIRTUAL BILLIARDS IN PSEUDO--EUCLIDEAN
SPACES]{VIRTUAL BILLIARDS IN PSEUDO--EUCLIDEAN
SPACES: DISCRETE HAMILTONIAN AND CONTACT INTEGRABILITY}

\author[Bo\v zidar Jovanovi\'c \and
Vladimir Jovanovi\'c]{}

\email{bozaj@mi.sanu.ac.rs}
\email{vlajov@blic.net}

\begin{document}

\maketitle

\centerline{\scshape Bo\v zidar Jovanovi\'c}
\medskip
{\footnotesize
\centerline{Mathematical Institute SANU}
\centerline{Serbian Academy of Sciences and Arts}
\centerline{Kneza Mihaila 36, 11000 Belgrade, Serbia}
}

\medskip

\centerline{\scshape Vladimir Jovanovi\'c}
\medskip
{\footnotesize
 \centerline{Faculty of Sciences}
   \centerline{University of Banja Luka}
   \centerline{Mladena Stojanovi\'ca 2, 51000 Banja Luka, Bosnia
and Herzegovina}
}

\bigskip

\begin{abstract}
The aim of the paper is to unify the efforts in the study of integrable billiards within quadrics in flat and curved spaces
and to explore further the interplay of symplectic and contact integrability.
As a starting point in this direction, we consider virtual billiard dynamics within quadrics in
pseudo--Euclidean spaces. In contrast to the usual
billiards, the incoming velocity and the velocity after the
billiard reflection can be at opposite sides of the tangent plane
at the reflection point. In the symmetric case we prove noncommutative integrability of the
system and give a geometrical interpretation of integrals, an
analog of the classical Chasles and Poncelet theorems and
we show that the virtual billiard dynamics provides a natural framework in
the study of billiards within quadrics in projective spaces, in
particular of billiards within ellipsoids on the sphere $\mathbb
S^{n-1}$ and the Lobachevsky space $\mathbb H^{n-1}$.
\end{abstract}

\tableofcontents

\section{Introduction}

It is well known that the billiards within ellipsoids are the only known integrable billiards with smooth boundary in constant
curvature spaces \cite{AA, Bo, BM1, BM2, DrRa, KoTr, Ta, Ta1, Ves3}. The elliptical billiards in pseudo-Euclidean spaces are also integrable \cite{KT, DR}.
We will try to present all these integrable models through a unified perspective, within the framework
of the virtual billiard dynamic (see \cite{JJ}).

A pseudo--Euclidean space $\mathbb E^{k,l}$ of signature $(k,l)$,
$k,l\in\N,\, k+l=n$, is the space $\R^{n}$ endowed with the scalar
product
\[
\langle x,y\rangle = \sum_{i=1}^k x_iy_i - \sum_{i=k+1}^{n}
x_iy_i\quad (x,y\in\R^{n}).
\]

Two vectors $x,y$ are {\it orthogonal}, if $\langle x,y\rangle=0$.
A vector $x\in \mathbb E^{k,l}$ is called \emph{space--like},
\emph{time--like}, \emph{light--like}, if $\langle x,x\rangle$
is positive, negative, or $x$ is orthogonal to itself,
respectively. Denote by $(\cdot,\cdot)$ the Euclidean inner
product in $\R^{n}$ and let
$$
E=\diag(\tau_1,\dots,\tau_n)=\mbox{diag}(1,\dots,1,-1,\dots,-1),
$$
where $k$ diagonal elements are equal to 1 and $l$ to $-1$. Then
$\langle x,y\rangle=(Ex,y)$, for all $x,y\in\R^{n}$.

We consider a $n-1$--dimensional quadric
\begin{equation}\label{elipsoid}
\mathbb{Q}^{n-1}=\left\{x\in \mathbb E^{k,l}\,|\, (A^{-1}x,x)=1 \right\},
\end{equation}
where
\begin{equation}\label{A-matrix}
A=\diag(a_1,\dots,a_n), \quad a_i\ne 0, \quad i=1,\dots,n.
\end{equation}

A point $x\in\mathbb{Q}^{n-1}$ is {\it singular}, if a normal $E\A x$ at $x\in\mathbb{Q}^{n-1}$
is
light--like: $(EA^{-2}x,x)=0$, or equivalently, the induced metric is degenerate at
$x$.

In the case that $A$ is positive definite, following Khesin and
Tabachnikov \cite{KT} and Dragovi\'c and Radnovi\'c \cite{DR}, we
define a billiard flow inside the ellipsoid \eqref{elipsoid} in
$\mathbb E^{k,l}$ as follows. Between the impacts, the motion is uniform
along the straight lines. If $x\in{\mathbb Q}^{n-1}$ is
non--singular, then the normal $E\A x$ is transverse to
$T_x\mathbb Q^{n-1}$ and the incoming velocity vector $w$ can be
decomposed as $w=t+n$, where $t$ is its tangential and $n$ the
normal component in $x$. The velocity vector after reflection is
$w_1=t-n$. If $x\in{\mathbb Q}^{n-1}$ is singular, the flow stops.

Let $\phi:(x_j,y_j)\mapsto (x_{j+1},y_{j+1})$ be the billiard
mapping, where $x_j\in{\mathbb Q}^{n-1}$ is a sequence of
non--singular impact points and $y_j$ is the corresponding
sequence of outgoing velocities (in the notation we follow
\cite{Ves3, Veselov, fedo}, which slightly differs from the one
given in \cite{MV}, where $y_{j}$ is the incoming velocity).
As in the Euclidean case (see \cite{Veselov, MV, fedo}), the billiard mapping $\phi$ is given by:
\begin{align}
x_{j+1}&=x_j+\mu_j y_j,\label{1bilijar}\\
y_{j+1}&=y_j+\nu_j EA^{-1}x_{j+1}, \label{2bilijar}
\end{align}
where the multipliers
$$
\mu_j=-2\,\frac{(\A x_j,y_j)}{(\A y_j,y_j)}, \qquad
\nu_j=2\,\frac{(\A
x_{j+1},y_{j+1})}{(EA^{-2} x_{j+1},x_{j+1})}
$$
are determined from the conditions
$$
(\A x_{j+1},x_{j+1})=(\A x_j,x_j)=1, \qquad \langle y_{j+1},y_{j+1}\rangle=\langle y_j,y_j\rangle.
$$

From the definition, the Hamiltonian $H=\frac12\langle y_j,y_j\rangle$
is an invariant of the mapping $\phi$. Therefore, the lines
$l_k=\{x_k+sy_k\,\vert\,s\in\R\}$ containing segments $x_kx_{k+1}$
of a given billiard trajectory are of the same type: they are all
either space--like ($H>0$), time--like ($H<0$) or light--like
($H=0$). Also, the function $J_j=(\A x_j,y_j)$ is an invariant of the
billiard mapping (see Lemma 3.1 in \cite{JJ}).

Note that the billiard mapping \eqref{1bilijar}, \eqref{2bilijar}
is well defined for arbitrary quadric $\mathbb Q^{n-1}$ given by
\eqref{elipsoid} and not only for ellipsoids. In that case, the outgoing velocity (directed from $x_k$ to $x_{k+1}$) is either $y_k$ or $-y_k$, while
the segments
$x_{k-1}x_k$ and $x_k x_{k+1}$ determined by 3 successive points
of the mapping \eqref{1bilijar}, \eqref{2bilijar} may be:

\begin{itemize}

\item[(i)] on the same side of the tangent plane $T_{x_k}\mathbb
Q^{n-1}$;

\item[(ii)] on the opposite sides of the tangent plane
$T_{x_k}\mathbb Q^{n-1}$.

\end{itemize}

\begin{figure}[ht]
\includegraphics[height=55mm]{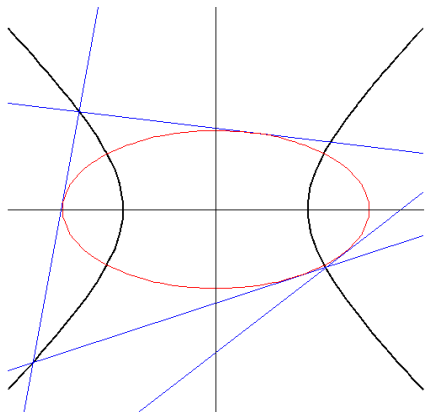}
\caption{A segment of a virtual billiard trajectory within
hyperbola ($a_1>0,a_2<0$) in the Euclidean space $\mathbb E^{2,0}$. The
caustic is an ellipse. }
\end{figure}

In the case (i) we have a part of the usual pseudo--Euclidean
billiard trajectory, while in the case (ii) the billiard
reflection corresponds to the points $x_{k-1} x_k x'_{k-1}$, where
$x'_{k+1}$ is the symmetric image of $x_{k+1}$ with respect to
$x_k$. In the three-dimensional Euclidean case, Darboux referred
to such reflection as the {\it virtual reflection} (e.g., see
\cite{DR2006} and \cite{DrRa}, Ch. 5). In  Euclidean spaces of
arbitrary dimension, such configurations were introduced by
Dragovi\'c and Radnovi\'c   in \cite{DR2006}. It appears that a
multidimensional variant of Darboux's 4--periodic virtual
trajectory with reflections on two quadrics, refereed as
double--reflection configuration \cite{DrRa}, is fundamental in
the construction of the double reflection nets in Euclidean spaces
(see \cite{DR2012}) and in pseudo-Euclidean spaces (see
\cite{DR2014}). They also played a role in a construction of the
billiard algebra in \cite{DR2008}. The 4--periodic orbits of real
and complex planar billiards with virtual reflections are also
studied in \cite{Gl}.

\begin{definition} \cite{JJ}
Let $\mathbb Q^{n-1}$ be a quadric in the pseudo--Euclidean space
$\mathbb E^{k,l}$ defined by \eqref{elipsoid}. We refer to
\eqref{1bilijar}, \eqref{2bilijar} as the {\it virtual billiard
mapping}, and to the sequence of points $x_k$ determined by
\eqref{1bilijar}, \eqref{2bilijar} as the {\it virtual billiard
trajectory} within $\mathbb Q^{n-1}$.
\end{definition}

The system is defined outside the singular set
\begin{equation}\label{singular}
\Sigma=\{(x,y)\in T\R^n\,\,\vert\,\,(EA^{-2}x,x)=0 \,\, \vee
(A^{-1}x,y)=0\,\,\vee\,\,(A^{-1} y,y)=0\}
\end{equation}
and it is invariant under the action of a discrete group $\mathbb
Z^{n}_2$ generated by the reflections
\begin{equation}\label{group}
(x_i,y_i)\,\longmapsto\, (-x_i,-y_i), \qquad i=1,\dots,n.
\end{equation}

We can interpret \eqref{1bilijar}, \eqref{2bilijar} in the case
of non--light--like billiard trajectories as the equations of a
discrete dynamical system (see \cite{Veselov, MV, Ves3}) on
$\mathbb Q^{n-1}$ described by the discrete action functional:
$$
S[\mathbf x]=\sum_k \mathbf{L}(x_k,x_{k+1}), \qquad \mathbf
L(x_k,x_{k+1})=\sqrt{\vert \langle
x_{k+1}-x_{k},x_{k+1}-x_{k}\rangle \vert },
$$
where $\mathbf x=(x_k), \, k\in\mathbb Z$ is a sequence of points
on $\mathbb Q^{n-1}$. Note that the virtual billiard dynamics on
$\mathbb Q^{n-1}$ can have both virtual and real reflections.

Motivated by the Lax reprezentation for elliptical billiards with the Hooke's potential (Fedorov \cite{fedo}, see also \cite{BozaRos, R}),
we proved in \cite{JJ} that the trajectories $(x_j,y_j)$ of
\eqref{1bilijar}, \eqref{2bilijar} outside the singular set
\eqref{singular} satisfy the matrix equation
\begin{equation}\label{LAbil}
\mathcal{L}_{x_{j+1},y_{j+1}}(\lambda)=\mathcal{A}_{x_j,y_j}(\lambda)\mathcal{L}_{x_j,y_j}(\lambda)\mathcal{A}_{x_j,y_j}^{-1}(\lambda),
\end{equation}
with $2\times2$ matrices depending on the parameter $\lambda$
\begin{align*}
\mathcal{L}_{x_j,y_j}(\lambda)&=\left(\begin{array}{cc}
q_{\lambda}(x_j,y_j) &
q_{\lambda}(y_j,y_j)\\-1-q_{\lambda}(x_j,x_j) &
-q_{\lambda}(x_j,y_j)
\end{array}\right),\\
\mathcal{A}_{x_j,y_j}(\lambda)&= \left(\begin{array}{cc}
I_j\lambda+2J_j\nu_j &-I_j\nu_j\\ -2J_j\lambda& I_j\lambda
\end{array}\right),
\end{align*}
where $q_{\lambda}$ is given by
\begin{equation}\label{bilin}
q_{\lambda}(x,y)=((\lambda
E-A)^{-1}x,y)=\sum_{i=1}^k\frac{x_iy_i}{\lambda-a_i}-\sum_{i=k+1}^n\frac{x_iy_i}{\lambda+a_i},
\end{equation}
and
\begin{equation}\label{POMOC}
J_j=(\A x_j,y_j),\quad I_j=-(\A
y_j,y_j),\quad\nu_j=2J_j/(EA^{-2}x_{j+1},x_{j+1}).
\end{equation}

For a non--symmetric case ($\tau_i a_i \ne \tau_ja_j$) the matrix
representation  is equivalent to the system up to the $\mathbb
Z_2^n$--action \eqref{group}. Further, from the expression
\begin{equation}\label{nonSymL}
\det \mathcal
L_{x,y}(\lambda)=q_\lambda(y,y)(1+q_\lambda(x,x))-q_\lambda(x,y)^2=\sum_{i=1}^n\frac{f_i(x,y)}{\lambda-\tau_ia_i},
\end{equation}
one can derive the integrals $f_i$ in the form
\begin{equation}\label{intgeo}
f_i(x,y)=\tau_i y_i^2+\sum_{j\neq
i}\frac{(x_iy_j-x_jy_i)^2}{\tau_j a_i-\tau_i a_j}\quad
(i=1,\dots,n).
\end{equation}

\subsection*{Outline and results of the paper.}
In Section 2 we describe discrete symplectic  (Theorem \ref{symplteo}) and
contact integrability in the light--like case (Theorem \ref{velika}) of the
virtual billiard dynamics directly, by the use of the
Dirac--Poisson bracket. This is slightly different from the
construction within the framework of the symplectic reduction
given by Khesin and Tabachnikov \cite{KT, KT2}.

In the symmetric case, when $a_i\tau_i=a_j\tau_j$ for some indexes
$i,j$, we further develop the analysis from \cite{JJ} of geodesic
flows on $\mathbb Q^{n-1}$ and elliptical billiards. We prove
noncommutative integrability of the system (Theorem \ref{integraliT}, Section
3) and, by a subtle estimate of the number of real zeros in the spectral parameter $\lambda$ of the rational function $\det\mathcal L_{x,y}(\lambda)$,  give a geometrical interpretation of integrals - an analog
of the classical Chasles and Poncelet theorems for symmetric
quadrics (Theorems \ref{ch1} -- \ref{poncelet}, Section 4). The Poncelet theorem is
based on a noncommutative variant of the description of Liouville
integrable symplectic correspondences given by Veselov \cite{Ves3,
Ves4} (Theorem \ref{DIS}, Section 3).

Further, in Section 5 we show that the virtual billiard dynamics
provides a natural framework in the study of billiards within
quadrics in projective spaces, in particular the billiards within
ellipsoids on the sphere  $\mathbb S^{n-1}$ and the Lobachevsky
space $\mathbb H^{n-1}$. It is well known that the ellipsoidal
billiards on $\mathbb S^{n-1}$ and $\mathbb H^{n-1}$  are
completely integrable \cite{Bo, Ves2, Ta, DGJ}. The "big" $n\times
n$--matrix representation of the ellipsoidal $\mathbb
H^{n-1}$--billiard, together with the integration of the flow is
obtained in \cite{Ves2}. In this paper we provide a "small"
$2\times 2$--matrix representation (Theorem \ref{billiardLA2}), a modification
of \eqref{LAbil}, as well as the Chasles theorem  (Theorem \ref{sfera}).

\section{Symplectic and contact properties of the virtual billiard dynamics}

\subsection{Hamiltonian description.}
In the pseudo-Euclidean case it is convenient to use the following
symplectic form on $\R^{2n}=T\mathbb E^{k,l}(x,y)$ (see \cite{KT}):
\[
\omega=Edy \wedge dx=\sum_{i=1}^k dy_i\wedge dx_i-\sum_{i=k+1}^n
dy_i\wedge dx_i,
\]
obtained after identification $T^*\mathbb E^{k,l}(x,p)\cong T\mathbb E^{k,l}(x,y)$
using the scalar product $\langle\cdot,\cdot\rangle$. The
corresponding Poisson bracket is
\begin{equation}\label{canonical}
\{f,g\}=\sum_{i=1}^k\frac{\p f}{\p x_i}\frac{\p g}{\p
y_i}-\sum_{i=k+1}^n\frac{\p f}{\p x_i}\frac{\p g}{\p y_i} -
\sum_{i=1}^k\frac{\p f}{\p y_i}\frac{\p g}{\p x_i} +
\sum_{i=k+1}^n\frac{\p f}{\p y_i}\frac{\p g}{\p x_i}.
\end{equation}

Consider a $(2n-2)$--dimensional submanifold $M_h$ of $\R^{2n}$
defined by
\begin{align*}
M_h &=\{(x,y)\in
\R^{2n}\backslash\Sigma\,\,\vert\,\,\phi_1=(A^{-1}x,x)=1, \,\,
\phi_2=2H=\langle y,y\rangle=h\}\\
   &=(\mathbb{Q}^{n-1}\times S^{n-1}_h)\backslash\Sigma,
\end{align*}
where $\Sigma$ is given by \eqref{singular} and
$S^{n-1}_h=\{y\in \R^n\,\vert\,\langle y,y\rangle=h\}$ is a
pseudosphere ($h\ne 0$) or a light--like cone ($h=0$).

Due to $\{\phi_1,\phi_2\}=4(A^{-1}x,y)\ne 0$ on $M_h$, it follows
that $M_h$ is a symplectic submanifold of $(\R^{2n},\omega)$.
Recall, for $F_1,F_2\in C^{\infty}(M_h)$, the Hamiltonian vector
field $X_{F_i}$ is defined by $i_{X_{F_i}}\omega_{M_h}=-dF_i$,
while the Poisson bracket is given by
$\{F_2,F_1\}_{M_h}=X_{F_1}(F_2)$.

Alternatively, we can define the Poisson bracket in redundant
variables by the use of Dirac's construction (e.g., see
\cite{Moser, Su}). Let $F_1=f_1\vert_{M_h}$, $F_2=f_2\vert_{M_h}$,
$f_1,f_2\in C^\infty(R^{2n})$. Then
\begin{equation}\label{PD}
\{F_1,F_2\}_{M_h}=\{f_1,f_2\}_{M_h}=\{f_1,f_2\}-\frac{\{\phi_1,f_1\}\{\phi_2,f_2\}-\{\phi_2,f_1\}\{\phi_1,f_2\}}{\{\phi_1,\phi_2\}}.
\end{equation}

The bracket is characterized by
\begin{equation}\label{symplbracket}
\{x_i,x_j\}_{M_h}=0,\quad \{x_i,y_j\}_{M_h}=\tau_i\delta_{ij}-\frac{x_jy_i\tau_j a_j^{-1}}{(A^{-1}x,y)}, \quad \{y_i,y_j\}_{M_h}=0.
\end{equation}

\begin{theorem}\label{symplteo}
{\rm (i)} The mapping $\phi:M_h\rightarrow M_h$,
$\phi(x_k,y_k)=(x_{k+1},y_{k+1})$ given by \eqref{1bilijar},
\eqref{2bilijar} is symplectic,
$$
\phi^{\ast}\omega_{M_h}=\omega_{M_h},
$$
where $\omega_{M_h}$ is
the restriction of the symplectic form $\omega$ to $M_h$.

{\rm (ii)} Assume that the quadric is not symmetric. The integrals
\eqref{intgeo} commute with respect to the Poisson bracket
$\{\cdot,\cdot\}_{M_h}$. The virtual billiard map is a completely
integrable discrete system on the phase space $M_h$, which is
almost everywhere foliated on $(n-1)$--dimensional Lagrangian
invariant manifolds.
\end{theorem}

\begin{proof} (i) Although it is straightforward, we feel
that it would be interesting to present a direct proof of the
statement. For our convenience we denote $x_k,y_k,\mu_k$,$\nu_k$,
$x_{k+1},y_{k+1}$ by $x,y,\mu,\nu,\tilde x,\tilde y$,
respectively. As earlier mentioned,
\begin{equation}\label{cons}
(A^{-1}\tilde x,\tilde y)=(A^{-1}x,y).
\end{equation}
Notice also that
\begin{equation}\label{minus}
(A^{-1}\tilde x,\tilde y)=-(A^{-1}\tilde x,y).
\end{equation}
Indeed, due to $\tilde y+y \in T_{\tilde x} \mathbb Q^{n-1}$, we
have
$$
(A^{-1}\tilde x,\tilde y)=\langle EA^{-1}\tilde x,\tilde
y\rangle=-\langle EA^{-1}\tilde x, y\rangle=-(A^{-1}\tilde x,y).
$$

According to \eqref{symplbracket} it suffices to prove that
\begin{equation}\label{symplbracketnov}
\{\tilde x_i,\tilde x_j\}_{M_h}=0,\quad \{\tilde x_i,\tilde y_j\}_{M_h}=\tau_i\delta_{ij}-\frac{\tilde x_j \tilde y_i\tau_j a_j^{-1}}{(A^{-1}\tilde x,\tilde y)}, \quad \{\tilde y_i,\tilde y_j\}_{M_h}=0.
\end{equation}

The proofs of the first and the third relation in
\eqref{symplbracketnov} are tedious and we will omit them here.
Assuming that $\{\tilde x_i,\tilde x_j\}_{M_h}=0$, we will prove
only the second relation. At the beginning let's show that
\begin{equation}\label{prva}
\{\tilde x_i,y_j\}_{M_h}=\tau_i\delta_{ij}-\frac{\tilde x_j y_i\tau_j a_j^{-1}}{(A^{-1}\tilde x, y)}.
\end{equation}

First, owing to $\{ y_i,y_j\}_{M_h}=0$ it is
\begin{align*}
\{(A^{-1}x,y),y_j\}_{M_h}=&\sum_{l=1}^ny_la_l^{-1}\{x_l,y_j\}_{M_h}\\
=&\sum_{l=1}^ny_la_l^{-1}\Big(\tau_l\delta_{lj}-
\frac{x_jy_l\tau_ja_j^{-1}}{(A^{-1}x,y)}\Big)\\
=& y_j\tau_ja_j^{-1}-\frac{x_j\tau_ja_j^{-1}}{(A^{-1}x,y)}\,(A^{-1}y,y).
\end{align*}
Consequently, from \eqref{symplbracket}, \eqref{cons}, \eqref{minus}, we have
\begin{align*}
\{\tilde x_i,y_j\}_{M_h}&= \{x_i-2\frac{(A^{-1}x,y)}{(A^{-1}y,y)}y_i,y_j\}_{M_h}\\
&= \{ x_i,y_j\}_{M_h}-\frac{2y_i}{(A^{-1}y,y)}\, \{(A^{-1}x,y),y_j\}_{M_h}\\
&= \tau_i\delta_{ij}-\frac{x_jy_i\tau_j a_j^{-1}}{(A^{-1}x,y)}-2\,\frac{y_iy_j\tau_j a_j^{-1}}{(A^{-1}y,y)}+2\,\frac{x_jy_i\tau_j a_j^{-1}}{(A^{-1}x,y)}\\
&= \tau_i\delta_{ij}+\frac{y_i\tau_j a_j^{-1}}{(A^{-1}x,y)}\Big(x_j-2\frac{(A^{-1}x,y)}{(A^{-1}y,y)}\, y_j\Big)\\
&= \tau_i\delta_{ij}-\frac{\tilde x_j y_i\tau_j a_j^{-1}}{(A^{-1}\tilde x, y)}.
\end{align*}
Now, using \eqref{prva} and \eqref{minus} we obtain
\begin{align*}
\{\tilde x_i,\nu\}_{M_h}&=\{\tilde x_i,2\frac{(A^{-1}\tilde x,\tilde y)}{(EA^{-2}\tilde x,\tilde x)}\}_{M_h}\\
&= -\frac2{(EA^{-2}\tilde x,\tilde x)}\{\tilde x_i, (A^{-1}\tilde x, y)\}_{M_h}\\
&=-\frac2{(EA^{-2}\tilde x,\tilde x)}\sum_{l=1}^n \tilde x_la_l^{-1}\{\tilde x_i, y_l\}_{M_h}\\
&= -\frac2{(EA^{-2}\tilde x,\tilde x)}\sum_{l=1}^n \tilde x_la_l^{-1}\Big(\tau_i\delta_{il}-\frac{\tilde x_l y_i\tau_l a_l^{-1}}{(A^{-1}\tilde x, y)}\Big)\\
&=  -\frac{2\tau_ia_i^{-1}\tilde x_i}{(EA^{-2}\tilde x,\tilde x)}+\frac{2y_i}{(A^{-1}\tilde x, y)}.
\end{align*}
Therefore,
\begin{align*}
\{\tilde x_i,\tilde y_j\}_{M_h} & =  \{\tilde x_i, y_j+\nu\tau_j a_j^{-1}\tilde x_j\}_{M_h}\\
&= \{\tilde x_i,y_j\}_{M_h}+\tau_j a_j^{-1}\tilde x_j \{\tilde x_i,\nu\}_{M_h}\\
&=  \tau_i\delta_{ij}-\frac{\tilde x_j y_i\tau_j a_j^{-1}}{(A^{-1}\tilde x, y)}
-2\,\frac{\tilde x_i\tilde x_j \tau_i a_i^{-1}\tau_j a_j^{-1}}{(EA^{-2}\tilde x,\tilde x)}
+2\,\frac{\tilde x_j y_i \tau_j a_j^{-1}}{(A^{-1}\tilde x, y)}\\
&=  \tau_i\delta_{ij}-\frac{\tilde x_j \tau_j
a_j^{-1}}{(A^{-1}\tilde x, \tilde y)}(y_i+\nu \tau_i
a_i^{-1}\tilde x_i)\\
&=\tau_i\delta_{ij}-\frac{\tilde x_j \tilde
y_i\tau_j a_j^{-1}}{(A^{-1}\tilde x,\tilde y)}.
\end{align*}

(ii) Note that the only relation between the integrals on $M_h$ is
\begin{equation}\label{relacija}
f_1+\dots+f_n=\langle y,y\rangle=h.
\end{equation}

Similarly as in the Euclidean space, we have $\{f_i,f_j\}=0$ (see
\cite{KT, KT2}). Further
$\{\phi_2,f_i\}=\{2H,f_i\}=\{f_1+\dots+f_n,f_i\}=0$, and therefore
$$
\{f_i,f_j\}_{M_h}=0, \quad i,j=1,\dots,n.
$$
\end{proof}

\begin{remark}{
Observe that $\{ \tilde x_i,\tilde x_j\}_{M_h}=0$, \eqref{prva},
and $\{ y_i,y_j\}_{M_h}=0$ imply that the mapping $(x,y)\mapsto
(\tilde x,y)$ is also symplectic on $M_h$.
}\end{remark}

\begin{remark}{
Note that in the virtual billiard mapping \eqref{1bilijar},
\eqref{2bilijar} we allow the trajectories both with $J>0$
and $J<0$ ($J=(A^{-1} x,y)=0$ defines the tangent space
$T_x\mathbb Q^{n-1}$). For example, in the ellipsoidal case when
$A$ is positive definite, $J>0$ means that $y$ is directed outward
$\mathbb Q^{n-1}$. It is also natural to consider the dynamics
of lines
$$
l_k=\{x_k+sy_k\,\vert\,s\in\R\}, \qquad k\in \mathbb Z,
$$
described by Khesin and Tabachnikov
within the framework of the symplectic reduction
for $A$ being positive definite \cite{KT}.
In our notation, in the space--like and time--like cases, the dynamics of lines corresponds to the virtual billiard dynamics
on $M_h/\pm 1$ with identified $y$ and $-y$, while in the light--like case it
corresponds to the induced dynamics on $\bar M=M_0/\R^*$, where we
take the projectivization of the light--like cone $S^{n-1}_0$.
The latter case will be studied in details below.
}\end{remark}

\subsection{Contact description.} In the light--like case $h=0$
we show  the existence of a contact structure associated to $M_0$.
Let us introduce an action of $\R^{\ast}=\R\backslash\{0\}$ on
$M_0$ by
$$
g_{\lambda}(x,y)=(x,\lambda y), \qquad \lambda\in\R^{\ast}.
$$
The action is evidently free and proper, from which we conclude
that the orbit space $\bar M:=M_0/\R^{\ast}$ is a smooth manifold
of dimension $\dim \bar M=\dim M_0-1=2n-3$ and the projection
$\pi:M_0\rightarrow \bar M$, $\pi(x,y)=(x,[y])$ is a surjective
submersion.

With the notation above, $(M_0,\omega_{M_0})$ is a symplectic
Liouville manifold:
$$
g_\lambda^*\omega_{M_0}=\lambda\omega_{M_0}.
$$

The associated Liouville vector field and the Liouville 1-form are
given by
$$
Z(x,y)=\frac{d}{d\lambda}g_{\lambda}(x,y)\big|_{\lambda=1}=(0,y)
$$
and
$$
\tilde\beta:=i_Z\omega_{M_0}=Ey\cdot dx \vert_{M_0},
$$
respectively. Then $d\tilde \beta=\omega_{M_0}$ and
$g_\lambda^*\tilde\beta=\lambda\tilde\beta$ (e.g, see \cite{LM}). It is well known that
the orbit space $\bar M$ carries the natural contact structure
induced by $\tilde\beta$ (Proposition 10.3, Ch. V, \cite{LM}). We
 describe this contact structure below.

Let
$$
\beta:=\frac{1}{J}\,\tilde\beta=\frac1{(A^{-1}x,y)}\,\tilde\beta.
$$

\begin{theorem}\label{velika}
{\rm (i)} There exists a unique 1-form $\bar\beta$ on $\bar M$, such
that $\beta=\pi^{\ast}\bar\beta$. Furthermore, the form
$\bar\beta$  is contact and $\bar R:=\pi_{\ast}X_J$ is the Reeb
vector field on $(\bar M,\bar\beta)$, where $X_J$ is the
Hamiltonian vector field of the function $J=(A^{-1}x,y)$ on $M_0$.

{\rm (ii)} The mapping $\bar\phi:\bar M\rightarrow\bar M$ defined by
$\bar\phi(x,[y]):=\pi(\phi(x,y))$ is contact,
$$
(\bar\phi)^{\ast}\bar\beta=\bar\beta.
$$

{\rm (iii)} Assume that the quadric is not symmetric. The functions
$f_i/J^2$ descend to the commutative integrals $\bar f_i$,
$$
[\bar f_i,\bar f_j]=0, \qquad i,j=1,\dots,n,
$$
of the contact mapping $\bar\phi$, where $[\cdot,\cdot]$ is the
Jacobi bracket on $(\bar M,\bar\beta)$. Further, $\bar f_i$ are
preserved by the Reeb vector field $\bar R$ of $(\bar
M,\bar\beta)$
$$
\bar R(\bar f_i)=0   \quad  \Longleftrightarrow \quad [1,\bar
f_i]=0, \quad i=1,\dots,n,
$$
and the contact mapping $\bar\phi$ is contact completely
integrable: the manifold $\bar M$ is almost everywhere foliated on
$(n-1)$--dimensional pre-Legendrian invariant manifolds.
\end{theorem}

\begin{proof} (i) We have, 
\begin{equation}\label{ker}
\mbox{ker}\,\pi_{\ast}=\mbox{span}\,\{Z\}.
\end{equation}

As a consequence of $g_\lambda^*\tilde\beta=\lambda\tilde\beta$
and $g_{\lambda}^* J=\lambda J$ we conclude that  $\beta$ is
$\R^*$--invariant, $g_{\lambda}^{\ast}\beta=\beta$. By definition
of $\beta$ it is $\beta(Z)=0$, which in view of \eqref{ker} implies
that $\tilde \beta$ is basic (e.g. see \cite{LM}, Ch.~II) and
there exists a unique 1-form $\bar\beta$ on $\bar M$, such that
$\beta=\pi^{\ast}\bar\beta$.

Further note
\begin{align*}
\beta\wedge(d\beta)^{n-2}&= \frac1J\,\tilde\beta\wedge\Big(\frac1J\, d\tilde\beta-\frac1{J^2}dJ\wedge\tilde\beta\Big)^{n-2}\\
&= \frac1{J^{n-1}}\,\tilde\beta\wedge(d\tilde\beta)^{n-2}\\
&= \frac1{J^{n-1}}\, (i_Z\omega_{M_0})\wedge\omega_{M_0}^{n-2}.
\end{align*}

Taking into account that $i_Z\omega_{M_0}^{n-1}=(n-1)(i_Z\omega_{M_0})\wedge\omega_{M_0}^{n-2}$, we obtain that
\begin{equation}\label{vol}
\beta\wedge(d\beta)^{n-2}=\frac1{(n-1)J^{n-1}}i_Z\omega_{M_0}^{n-1}.
\end{equation}

Let $\bar\gamma_1,\dots,\bar\gamma_{2n-3}\in T_{(x,[y])}\bar M$ be
arbitrary linearly independent tangent vectors. Since $\pi$ is a
submersion, there exist $\gamma_1,\dots,\gamma_{2n-3}\in T_{(x,y)}
M_0$, such that $\pi_{\ast}\gamma_i=\bar\gamma_i$, for all
$i=1,\dots,2n-3$. According to \eqref{ker}, the vectors
$Z,\gamma_1,\dots,\gamma_{2n-3}$ are linearly independent. Because
$\omega_{M_0}^{n-1}$ is a volume form on $M_0$, from \eqref{vol}
we have
\begin{align*}
\bar\beta\wedge(d\bar\beta)^{n-2}(\bar\gamma_1,\dots,\bar\gamma_{2n-3})&=
\beta\wedge(d\beta)^{n-2}(\gamma_1,\dots,\gamma_{2n-3})\\
&=\frac1{(n-1)J^{n-1}}\,\omega_{M_0}^{n-1}(Z,\gamma_1,\dots,\gamma_{2n-3})\neq0.
\end{align*}

Hence, $\bar\beta$ is a contact form on $\bar M$.

Now, let $X_J$ be the Hamiltonian vector field of $J$ on $M_0$. We
have
\[
\tilde\beta(X_J)=\omega_{M_0}(Z,X_J)=dJ(Z)=\sum_{i=1}^na_i^{-1}(x_idy_i+y_idx_i)(Z)=J.
\]
Consequently,
$$
\bar\beta(\bar
R)=\bar\beta(\pi_{\ast}X_J)=\beta(X_J)=\frac1J\,\tilde\beta(X_J)=1
$$
and $\bar R:=\pi_{\ast}X_J$ is the Reeb vector field on $\bar M$.

\medskip

(ii) Evidently, $g_{\lambda}\circ\phi=\phi\circ g_{\lambda}$ for
all $\lambda\in\R^{\ast}$ and $\bar\phi$ is well defined. Taking
derivative in $\lambda=1$, we get $\phi_{\ast}Z=Z$ and
$i_Z\phi^{\ast}\omega_{M_0}=\phi^{\ast}(i_Z\omega_{M_0})$.
According to Theorem \ref{symplteo} the symplectic form
$\omega_{M_0}$ is $\phi$--invariant,
$\phi^{\ast}\omega_{M_0}=\omega_{M_0}$, and consequently,
\[
\phi^{\ast}\tilde\beta=\phi^{\ast}(i_Z\omega_{M_0})=i_Z\phi^{\ast}\omega_{M_0}=i_Z\omega_{M_0}=\tilde\beta.
\]
Dividing the last equation by $J$ and using $\phi^{\ast}J=J$, we get $\phi^{\ast}\beta=\beta$. This implies that
\[
\pi^{\ast}(\bar\phi)^{\ast}\bar\beta=(\bar\phi\circ\pi)^{\ast}\bar\beta=(\pi\circ\phi)^{\ast}\bar\beta=
\phi^{\ast}\pi^{\ast}\bar\beta=\phi^{\ast}\beta=\beta=\pi^{\ast}\bar\beta.
\]
Using the fact that $\pi$ is a submersion, we finally obtain
$(\bar\phi)^{\ast}\bar\beta=\bar\beta$.

\medskip

(iii) The Jacobi brackets $[\bar f_i,\bar f_j]$  are given by
$$
[\bar f_i,\bar f_j]=\bar Y_{\bar f_i}\bar f_j-\bar f_j\,\bar R
\bar f_i, \qquad i,j=1,\dots,n,
$$
where $\bar R$ is the Reeb vector field on $(\bar M,\bar\beta)$,
$\bar\beta(\bar R)=1$, $i_{\bar R} d\bar\beta=0$, and
$$
\bar Y_{\bar f_i}=\bar f_i\bar R+\bar H_i, \qquad i=1,\dots,n,
$$
is the contact Hamiltonian vector field of $\bar f_i$. Here,
$\bar H_i$ are the horizontal vector fields, $\bar\beta(\bar
H_i)=0$, satisfying
\begin{equation}\label{barH}
d\bar\beta(\bar H_i,\bar X)=-\big(d\bar f_i(\bar X)-\bar R\bar f_i\,\bar\beta(\bar X)\big), \qquad
i=1,\dots,n,
\end{equation}
for all tangent vectors $\bar X$ on $\bar M$.

In addition, having in mind that each tangent vector $\bar X$ on
$\bar M$ has the form $\bar X=\pi_{\ast}X$ for some vector field
$X$ on $M_0$, we have
\begin{align*}
d\bar\beta(\bar X,\bar R)&= d\bar\beta(\pi_{\ast}X,\pi_{\ast}X_J)= d\beta(X,X_J)\\
&=\frac1J\, \omega_{M_0}(X,X_J)-\frac1{J^2}\,( dJ\wedge\tilde\beta)(X,X_J)\\
&= \frac1J\big[dJ(X)-\frac1J\, dJ(X)\tilde\beta(X_J)+\frac1J\, dJ(X_J)\tilde\beta(X)\big]\\
&= \frac1J\big[dJ(X)-\frac1J\, dJ(X)\,J\big]=0.
\end{align*}

Next, we prove that $\bar f_i$ are integrals of the Reeb vector field $\bar R$. As the first step we need the assertion
\begin{equation}\label{Jkom}
\{J,f_i\}_{M_0}=0,
\end{equation}
for all integrals $f_i$, which, for example, follows from
\eqref{veza}. Using this, from the definition
${f_i}/{J^2}=\pi^{\ast}\bar f_i$, we have
 \begin{align}
 \bar R\bar f_i &= d\bar f_i(\pi_{\ast}X_J)\nonumber\\
 &= d\big(\frac{f_i}{J^2}\big)(X_J)\nonumber\\
 &= \frac1{J^2}\, df_i(X_J)-\frac{2f_i}{J^3}\, dJ(X_J)\label{Rfi}\\
 &= \frac1{J^2}\,\{J,f_i\}_{M_0}= 0.\nonumber
 \end{align}

There exist, at least locally, vector fields $H_i$ that project
to horizontal vector fields $\bar H_i$: $\pi_{\ast}H_i=\bar H_i$.
If we substitute $\bar X=\pi_{\ast}X_{f_j}$ in \eqref{barH}, we
obtain
\begin{equation}\label{H}
d\beta(H_i, X_{f_j})=-d\Big(\frac{f_i}{J^2}\Big)(X_{f_j}).
\end{equation}
 Our aim is to prove that
\begin{equation}\label{rel}
df_j(H_i)=\frac{2f_j}{J}\, dJ(H_i).
\end{equation}
Due to
\begin{align*}
d\Big(\frac{f_i}{J^2}\Big)(X_{f_j})&=\frac1{J^2}\, df_i(X_{f_j})-
\frac{2f_i}{J^3}\, dJ(X_{f_j})\\
&=\frac1{J^2}\,
\{f_i,f_j\}_{M_0}-\frac{2f_i}{J^3}\,\{J,f_j\}_{M_0}=0,
\end{align*}
the relation \eqref{H} becomes $d\beta(H_i,X_{f_j})=0$, or equivalently,
\begin{equation}\label{equa}
d\tilde\beta(H_i,X_{f_j})=\frac1J\,(dJ\wedge\tilde\beta)(H_i,X_{f_j}).
\end{equation}

Owing to
\[
\tilde\beta(X_{f_j})=\omega_{M_0}(Z,X_{f_j})=df_j(Z)=2f_j,
\]
and using \eqref{Jkom}, we obtain
\[
\frac1J\,(dJ\wedge\tilde\beta)(H_i,X_{f_j})=
\frac1J\,\big[dJ(H_i)\tilde\beta(X_{f_j})
-dJ(X_{f_j})\tilde\beta(H_i)\big]=\frac{2f_j}{J}\, dJ(H_i).
\]

On the other hand
\[
d\tilde\beta(H_i,X_{f_j})=\omega_{M_0}(H_i,X_{f_j})=df_j(H_i),
\]
which together with \eqref{equa} yields \eqref{rel}. In the end,
thanks to \eqref{Rfi}, \eqref{rel} we have
\begin{align*}
[\bar f_i,\bar f_j]&= \bar Y_{\bar f_i}\bar f_j-\bar f_j\,\bar R \bar f_i\\
&= d\bar f_j(\bar H_i)\\
&= d\Big(\frac{f_j}{J^{2}}\Big)(H_i)\\
&= \frac1{J^2}\, df_j(H_i)-\frac{2f_j}{J^3}\, dJ(H_i)=0.
\end{align*}

Finally note that the integrals $f_i$ and $J$ on $M_h$ are related
by
\begin{equation}\label{veza}
J^2=\sum_{i=1}^n \tau_ia_i^{-1} f_i,
\end{equation}
which together with \eqref{relacija} imply that among the
integrals $\bar f_i$ we have two relations,
$$
\bar f_1+\dots+\bar f_n=0, \qquad \tau_1a_1^{-1} \bar
f_1+\dots+\tau_n a_n^{-1} \bar f_n=1,
$$
and that the number of the independent ones is $n-2$. According to
the  theorem on contact integrability, their invariant level-sets
almost everywhere define $(n-1)$--dimensional pre-Legendrian
manifolds, which have an additional $(n-2)$--dimensional
Legendrian foliation (see \cite{KT2, Jov}).
\end{proof}

\section{Noncommutative integrability and symmetric quadrics}

\subsection{Discrete noncommutative integrability.}

Recall that a Hamiltonian flow on a $2n$-dimensional symplectic
manifold $(M^{2n},\omega)$ (respectively, a contact flow on a
$2n+1$-dimensional contact manifold $(M^{2n+1},\beta)$) is {\it
noncommutatively integrable}, if it has a {\it complete set of
integrals} $\mathcal F$. The set $\mathcal F$ closed under the
Poisson bracket (respectively, the Jacobi bracket) is {\it
complete}, if one can find $2n-r$ almost everywhere independent
integrals $F_1,F_2,\dots,F_{2n-r}\in\mathcal F$, such that
$F_1,\dots,F_r$ Poisson commute with all integrals \cite{N,MF}
(respectively, $F_1,\dots,F_r$ commute with respect to the Jacobi
bracket with all integrals, and the functions in $\mathcal F$ are
integrals of the Reeb flow, as well \cite{Jov}).

Regular compact connected invariant manifolds of the system are
$r$-dimensional isotropic tori generated by the Hamiltonian flows
of $F_1,\dots,F_r$, i.e., $r+1$-dimensional pre-isotropic tori
generated by the Reeb vector field and the contact Hamiltonian
flows of $F_1,\dots,F_r$. Here, a submanifold $N\subset M^{2n+1}$
is {\it pre-isotropic}, if it transversal to the contact
distribution $\mathcal H=\ker\beta$ and if $\mathcal G_x=T_x N
\cap\mathcal H_x$ is an isotropic subspace of the symplectic
linear space $(\mathcal H_x,d\beta)$, for all $x\in N$. The last
condition is equivalent to the condition that distribution
$\mathcal G=\bigcup_x G_x$ defines a foliation \cite{Jov}.

In a neighborhood of a regular torus there exist canonical {\it
generalized action--angle coordinates} \cite{N} ({\it generalized
contact action---angle coordinates} \cite{Jov}), such that
integrals $F_i$, $i=1,\dots,r$ depend only on the actions and the flow
is a translation in the angle coordinates. If $r=n$ we have the usual
Liouville integrability described in the Arnold-Liouville theorem
\cite{Ar}, i.e., contact integrability described in \cite{BM,
KT2}.

If instead of the continuous flow we consider the symplectic
mapping $\Phi: M^{2n}\to M^{2n}$, $\Phi^*\omega=\omega$ (the
contact mapping $\Phi: M^{2n+1}\to M^{2n+1}$, $\Phi^*\beta=\beta$)
having the complete set of integrals $\mathcal F$, as above,
compact connected components of an invariant regular level set
\begin{equation}\label{Mc}
M_c=\{ F_1=c_1, \, F_2=c_2,\, \dots\, ,\,F_{2n-r}=c_{2n-r}\,\}
\end{equation}
are $r$-dimensional isotropic tori ($r+1$-dimensional
pre-isotropic tori) and in their neighborhoods there exist canonical
generalized (contact) action--angle coordinates.

By the same argumentation as given by Veselov \cite{Ves3, Ves4}
for the Liouville integrable symplectic correspondences, we have
the following description of the dynamics.

\begin{theorem}\label{DIS}
Let $M_c=T_1 \cup T_2 \cup \dots \cup T_p$ be a compact regular
level set \eqref{Mc}. If the torus $T_i \cong \mathbb
R^{r(+1)}/\Lambda_i$ is $\Phi$--invariant, then the restriction of
the mapping $\Phi$ to $T_i$ is the shift by a constant vector
$a_i\in \mathbb R^{r(+1)}$
$$
\Phi([x])\equiv x+a_i, \qquad [x]\in T_i.
$$
Otherwise, if
$$
\Phi(T_{i_k})=T_{i_{k+1}}, \quad k=1,\dots,q\le p,\quad
i=i_1=i_{q+1}, \quad T_{i_k} \cong \mathbb
R^{r(+1)}/\Lambda_{i_k},
$$
define tori $T_{i_k i_{k+1}}=\mathbb R^{r(+1)}/\Lambda_{i_k
i_{k+1}}$ by the lattices
$$
\Lambda_{i_k i_{k+1}}=\{ b\in \mathbb R^{r(+1)}\,\vert\,
\Phi([x+b])\equiv \Phi([x])\}=\{ b\in \mathbb R^{r(+1)}\,\vert\,
\Phi([x])\equiv \Phi([x])+b\},
$$
$[x]\in T_{i_k}, \Phi([x])\in T_{i_{k+1}}$, containing
$\Lambda_{i_k}$ and $\Lambda_{i_{k+1}}$ as sublattices. Then we
have the following commutative diagrams
\begin{equation*} \begin{CD}
 T_{i_k}  @ > \Phi >> T_{i_{k+1}} \\
 @ V  \pi_{i_k}  VV   @ VV \pi_{i_{k+1}}   V \\
 T_{i_k i_{k+1}}  @ >  \tau_{a_{i_k i_{k+1}}}  >>  T_{i_k i_{k+1}} \end{CD}
\end{equation*}
where $\tau_{a_{i_k i_{k+1}}}$ are the shifts by constant vectors
$a_{i_k i_{k+1}}\in\mathbb R^{r(+1)}$. The $q$-th iteration of
$\Phi$ is given by
$$
\Phi^q([x])\equiv x+a_{i_k},\qquad [x]\in T_{i_k},
$$
for some vectors $a_{i_k}\in \mathbb R^{r(+1)}$. In particular, if
a point $[x]\in T_{i_k}$ is periodic with a period $mq$, then all
points of $T_{i_1}\cup T_{i_2} \cup \dots \cup T_{i_q}$ are
periodic with the same period.
\end{theorem}

\subsection{Symmetric quadrics.}
We turn back to the virtual billiard dynamics and consider the
case when the quadric $\mathbb Q^{n-1}$ is symmetric. Define the
sets of indices $I_s\subset\{1,\dots,n\}\enspace (s=1,\dots r)$ by
the conditions
%
%
\begin{equation}\label{sym}
\begin{array}{l}
1^{\circ}\enspace \tau_i a_i=\tau_j a_j=\alpha_s\enspace \mbox{for}\enspace i,j\in I_s\enspace\mbox{and for all}\enspace s\in\{1,\dots,r\},\\[1ex]
2^{\circ}\enspace \alpha_s\neq \alpha_t\enspace \mbox{for}\enspace
s\neq t.
\end{array}
\end{equation}

Let
$$
\mathbb E^{k,l}=\mathbb E^{k_1,l_1} \oplus \dots\oplus \mathbb E^{k_r,l_r}
$$
be the associated decomposition of $\mathbb E^{k,l}$, where $\mathbb E^{k_s,l_s}$
are pseudo--Euclidean subspaces of the signature $(k_s,l_s)$ with
\[
k_s=\vert \{\tau_i\,\vert\,\tau_i=1,\,i\in I_s\}\vert, \quad
l_s=\vert \{\tau_i\,\vert\,\tau_i=-1,\,i\in I_s\}\vert, \quad
k_s+l_s=\vert I_s\vert.
\]

By $\langle \cdot,\cdot\rangle_s$ we denote the restriction of the
scalar product to the subspace $\mathbb E^{k_s,l_s}$:\footnote{To simplify
the notation, we omitted the projection operator $\pi_{s}:
\mathbb E^{k,l}\to \mathbb E^{k_s,l_s}$ at the left hand side of \eqref{scalar}.}
\begin{equation}\label{scalar}
\langle x,x\rangle_s=\sum_{i\in I_s}\tau_i x_i^2, \qquad x\in
\mathbb E^{k,l}.
\end{equation}

Let $SO(k_s,l_s)$ be the special orthogonal group of
$\mathbb E^{k_s,l_s}$. The quadric, as well as the virtual billiard flow,
is $ SO(k_1,l_1)\times\dots\times SO(k_r,l_r)$--invariant. The
integrals
\begin{equation}\label{moment-map}
\Phi_{s,ij}:=y_ix_j-x_iy_j,\qquad i,j\in I_s
\end{equation}
are proportional to the components of the corresponding momentum
mapping
$$
\Phi\colon M_h \longrightarrow so(k_1,l_1)^*\times \dots\times
so(k_r,s_r)^*.
$$

On the other hand, the determinant $\mbox{det}\,{\mathcal
L_{x,y}}(\lambda)$ is an invariant of the flow, and by expanding
it in terms of $1/(\lambda-\alpha_s),1/(\lambda-\alpha_s)^2$, we
get
\begin{align}\label{SymL}
\det\mathcal L_{x,y}(\lambda) &=
(1+q_{\lambda}(x,x))q_\lambda(y,y)-q_\lambda(x,y)^2\\
\nonumber &=\sum_{s=1}^r \frac{F_s}{\lambda-\alpha_s}+\frac{P_s}{(\lambda - \alpha_s)^2},
\end{align}
where the integrals $F_s, P_s$ are given by: \footnote{In
\cite{JJ} the term $\tau_i\tau_j$ is
omitted in the formula for $P_s$. This misprint, however, does not affect the results in \cite{JJ}. }
%
\begin{align*}
\label{integrals}F_s  =&\ds\sum_{i\in I_s}\Big(\tau_i y_i^2+\sum_{j\notin I_s}\frac{(x_iy_j-x_jy_i)^2}{\tau_j a_i-\tau_i a_j}\Big),\\
\nonumber P_s  =&\sum_{i,j\in I_s,i<j}\tau_i\tau_j\Phi^2_{s,ij}\quad
\text{for}\quad |I_s|\ge2 \qquad (P_s\equiv 0, \quad
\text{for}\quad |I_s|=1).
\end{align*}
%

The Hamiltonian is equal to the sum $H=\frac12\sum_{s=1}^r F_s$,
that is, among integrals $F_s$ we have the relation $\sum_s
F_s=2h$ on $M_h$.

For $h=0$, by $\bar F_s, \bar P_s, \bar \Phi_{s,ij}$ we denote the
functions on $\bar M$ obtained from $\R^*$--invariant integrals
$F_s/J^2, P_s/J^2, \Phi_{s,ij}/J$.

\begin{theorem}\label{integraliT}
{\rm (i)} The virtual billiard flow within symmetric quadric
\eqref{elipsoid}, \eqref{sym} is completely integrable in a
noncommutative sense by means of integrals $\mathcal F=\{F_s,
\Phi_{s,ij}\}$. The functions $F_s$,
$P_s=\sum_{i<j}\tau_i\tau_j\Phi_{s,ij}^2$ are central within the
algebra of integrals generated by $\mathcal F$:
\begin{eqnarray*}
&\{F_s,F_t\}_{M_h}=0,\quad \{F_s,P_t\}_{M_h}=0,\quad
\{P_s,P_t\}_{M_h}=0,\\
& \{F_s,\Phi_{t,ij}\}_{M_h}=0,\quad \{P_s,\Phi_{t,ij}\}_{M_h}=0,
\end{eqnarray*}
and their Hamiltonian vector fields generate $N-1$--dimensional
isotropic manifolds, regular level sets of the integrals $\mathcal
F$, where
$$
N=r+|\{s\in\{1,\dots,r\}\,:\, \vert I_s\vert\ge 2\}|.
$$

{\rm (ii)} In the light--like case, the mapping $\bar\phi$ is contact
completely integrable in a noncommutative sense by means of
integrals $\bar{\mathcal F}=\{\bar F_s, \bar{\Phi}_{s,ij}\}$. The
integrals are invariant with respect to the Reeb flow
$$
[1,\bar F_s]=0, \quad [1,\bar P_s]=0, \quad [1,\bar\Phi_{s,ij}]=0,
$$
and the functions $\bar F_s$, $\bar P_s$ are central within the
algebra of integrals generated by $\bar{\mathcal F}$:
\begin{eqnarray*}
& [\bar F_s,\bar F_t]=0,\quad [\bar F_s,\bar P_t]=0,\quad [\bar
P_s,\bar P_t]=0,\\
& [\bar F_s,\bar{\Phi}_{t,ij}]=0,\quad [\bar
P_s,\bar{\Phi}_{t,ij}]=0.
\end{eqnarray*}
Among central functions $\bar F_s$, $\bar P_s$ there are
$(N-2)$--independent ones and their contact Hamiltonian vector
fields, together with the Reeb vector field $\bar R$, generate
$N-1$--dimensional pseudo--isotropic manifolds  --  regular levels
sets of the integrals $\bar{\mathcal F}$.
\end{theorem}

The first statement is an analog of Theorems 5.1, 5.2 for the the
Jacobi-Rosochatius problem \cite{BozaRos} and Theorem 4.1 for
geodesic flows on quadrics in pseudo--Euclidean spaces \cite{JJ},
where the Dirac construction is applied for the constraints
$$
(A^{-1}x,x)=1, \qquad (A^{-1}x,y)=0.
$$

The second statement follows from the same considerations as in
the proof of Theorem \ref{velika}.
For example, similarly as in
\eqref{Rfi}, we have
\begin{align*}
 \bar R\bar{\Phi}_{s,ij} &= d\bar{\Phi}_{s,ij}(\pi_{\ast}X_J) =d\big(\frac{\Phi_{s,ij}}{J}\big)(X_J)      \\
 &= \frac1{J}\, d\Phi_{s,ij}(X_J)-\frac{\Phi_{s,ij}}{J^2}\, dJ(X_J)=\frac1{J}\,\{J,\Phi_{s,ij}\}_{M_0}= 0.
 \end{align*}
The last equality follows from the commuting relations
$\{J,\phi_2\}=0$, $\{\Phi_{s,ij},\phi_2\}=0$, and
$\{J,\Phi_{s,ij}\}=0$.

Note that the relation \eqref{SymL} for $\lambda=0$ implies $
J^2=\sum_s (\alpha_s^{-1} F_s-\alpha_s^{-2} P_s)$, whence the relations
$$
\sum_s \bar F_s=0, \qquad \sum_s (\alpha_s^{-1} \bar
F_s-\alpha_s^{-2} \bar P_s) = 1
$$
among the integrals $\bar F_s, \bar P_s$ on $\bar M$.

\begin{remark}{
An example of noncommutatively integrable multi-valued symplectic correspondence
is a recently constructed discrete Neumann system on a Stiefel variety  \cite{FeJo}.
Another example of a discrete integrable contact system is the Heisenberg model in pseudo--Euclidean
spaces \cite{Jov2}. We shall discus relationship between the Heisenberg model and virtual billiard dynamics in a forthcoming paper.
}\end{remark}

\section{The Chasles and Poncelet theorems for symmetric quadrics}

\subsection{Pseudo--confocal quadrics.}
There is a nice geometric manifestation of integrability of
elliptical billiards in pseudo--Euclidean spaces given by Khesin
and Tabachnikov \cite{KT}. Consider the following
"pseudo--confocal" family of quadrics in $\mathbb E^{k,l}$
\begin{equation}\label{confocal}
\mathcal Q_\lambda: \quad  ((A-\lambda E)^{-1} x,
x)=\sum_{i=1}^n\frac{x_i^2}{a_i-\tau_i\lambda}=1, \quad \lambda
\ne \tau_i a_i, \quad i=1,\dots,n.
\end{equation}
For a nonsymmetric ellipsoid, the lines $l_k$, $k\in\mathbb Z$
determined by a generic space--like or time--like (respectively
light--like) billiard trajectory are tangent to $n-1$
(respectively $n-2$) fixed quadrics from the pseudo--confocal
family \eqref{confocal} ({\it pseudo--Euclidean version of the
Chasles theorem}, see Theorem 4.9 in \cite{KT} and Theorem 5.1 in
\cite{DR}). A related geometric structure of the set of singular
points for the pencil \eqref{confocal} is described in \cite{DR,
DR2014}.

Here we consider the case of symmetric quadrics and further
develop the analysis given in \cite{JJ}, where $A$ had been positive
definite.

Without loss of generality we assume in the section that
\begin{equation}\label{alfici}
\alpha_1>\alpha_2>\dots>\alpha_r.
\end{equation}
The equation \eqref{confocal}
has $r$ solutions in the complex plane for a generic $x$. The
following lemma estimates the number of real solutions in certain
cases.

\begin{lemma}\label{jacobi}
{\rm (i)} Through points $x\in \mathbb E^{k,l}$ that satisfy
\begin{eqnarray}
&& \mathrm{sign} \langle x,x \rangle_s=\kappa_1\ne 0, \quad
s=1,\dots,g,\nonumber\\
&& \mathrm{sign} \langle x,x \rangle_{s}=\kappa_2\ne 0, \quad
s=g+1,\dots,r,\label{znak1}
\end{eqnarray}
for some index $g$ pass either $r$ quadrics (when $\kappa_1=-1,
\kappa_2=+1$, $\kappa_1=\kappa_2=+1$ or $\kappa_1=\kappa_2=-1$),
or $r$ resp. $r-2$ quadrics (when $\kappa_1=+1$, $\kappa_2=-1$)
from the pseudo--confocal family \eqref{confocal}. Similarly, if
\begin{eqnarray}
&&\mathrm{sign} \langle x,x \rangle_s=\kappa_1, \quad s=1,\dots,g_1,g_2,\dots,r,\nonumber \\
&& \mathrm{sign} \langle x,x \rangle_{s}=\kappa_2,\quad
s=g_1+1,\dots,g_2-1,\quad \kappa_1\cdot\kappa_2=-1,\label{ZNAK}
\end{eqnarray}
for some indexes $g_1,g_2$, $g_1<g_2$, through $x$ pass
either $r$ or $r-2$ quadrics from the pseudo--confocal family
\eqref{confocal}.

{\rm (ii)} The quadrics passing through arbitrary point $x$ are
mutually orthogonal at $x$.
\end{lemma}

\begin{proof} (i) We slightly modify the proof of the
corresponding Khesin and Tabachnikov statement given for
non-symmetric ellipsoids (Theorem 4.5 \cite{KT}). Consider the
function
$$
S(\lambda)=((A-E\lambda)^{-1} x, x)=\sum_{s=1}^r \frac{\langle
x,x\rangle_s}{\alpha_s-\lambda}.
$$

We have
\begin{eqnarray*}
&& S(\lambda) \, \sim \, -1/\lambda\langle x,x\rangle , \qquad
\lambda\rightarrow \pm\infty,\\
&& S(\lambda)\, \sim\, \frac{\langle
x,x\rangle_{s}}{\alpha_{s}-\lambda}, \qquad \lambda\rightarrow
\alpha_{s}, \qquad s=1,\dots, r,
\end{eqnarray*}
implying
\begin{equation*}
  \lim_{\lambda\to\pm\infty}S=0\,\quad \lim_{\lambda\to\alpha_{s}-}S=\mathrm{sign} \langle x,x
\rangle_s\cdot \infty\,
 ,\quad
\lim_{\lambda\to\alpha_{s}+}S=-\mathrm{sign} \langle x,x
\rangle_s\cdot\infty \, .
\end{equation*}

Therefore, if \eqref{znak1} holds, the equation $S(\lambda)=1$ has
real solutions in the $r-2$ intervals $(\alpha_{s+1},\alpha_{s})$,
$s=1,\dots,r-1$, $s\ne g$. In addition, we also have 2 real
solutions for $\kappa_1=-1$, $\kappa_2=+1$ (in the intervals
$(-\infty,\alpha_r)$, $(\alpha_1,\infty)$) and in the case when
all signs are equal (in the intervals $(-\infty,\alpha_r)$,
$(\alpha_{g+1},\alpha_{g})$, for $\kappa_1=\kappa_2=+1$, and in
the intervals $(\alpha_1,\infty)$, $(\alpha_{g+1},\alpha_{g})$,
for $\kappa_1=\kappa_2=-1$).

In the case when \eqref{ZNAK} holds, the equation $S(\lambda)=1$
always has real solutions in the $r-3$ intervals
$(\alpha_{s+1},\alpha_{s})$, $s=1,\dots,r-1$, $s\ne g_1,g_2-1$,
and an additional solution in the interval $(\alpha_1,\infty)$
for $\kappa_1=-1$, $\kappa_2=+1$, i.e., in the interval
$(-\infty,\alpha_r)$ for $\kappa_1=+1$, $\kappa_2=-1$.

\medskip

(ii) The second statement has the same proof as in the case
when $A$ is positive definite (Theorem 4.5 \cite{KT}).
\end{proof}

\begin{exm}{
From Lemma \ref{jacobi} it follows that in the Euclidean space
$\mathbb E^{n,0}$ through a generic point pass $r$ quadrics, while
through a generic point in the Lorentz--Poincar\'e--Minkowski
space $\mathbb E^{n-1,1}$ pass $r$ or $r-2$ quadrics from the
pseudo--confocal family \eqref{confocal} for arbitrary symmetric
quadric $\mathbb Q^{n-1}$ (Figures 2 and 3).
}\end{exm}

\begin{figure}[ht]
\includegraphics[height=55mm]{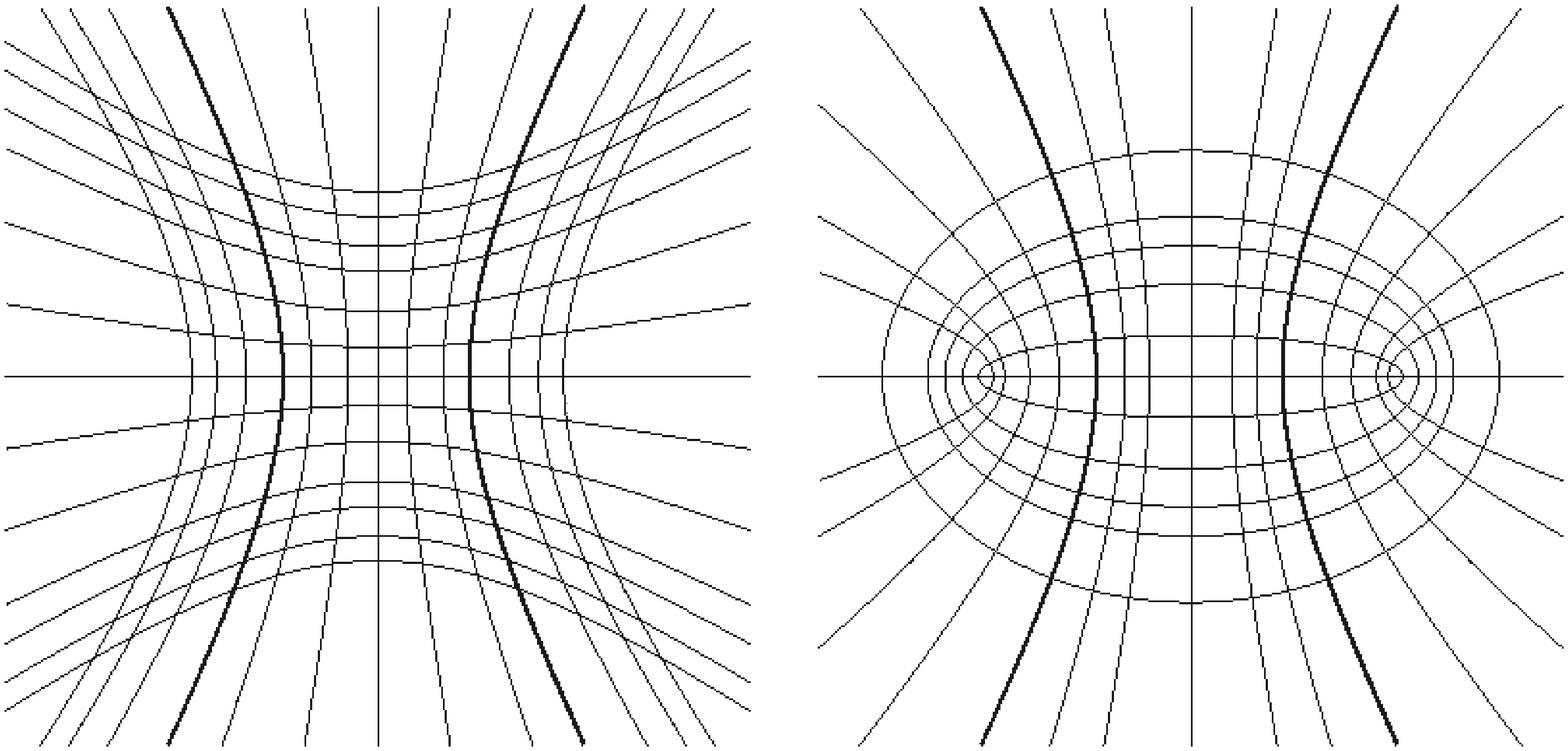}
\caption{Families of pseudo-confocal quadrics for $a_1>0,a_2<0$
in $\mathbb E^{1,1}$ (with $\alpha_1=-a_2>\alpha_2=a_1$) and $\mathbb E^{2,0}$,
respectively.}
\end{figure}

\begin{figure}[ht]
\includegraphics[height=55mm]{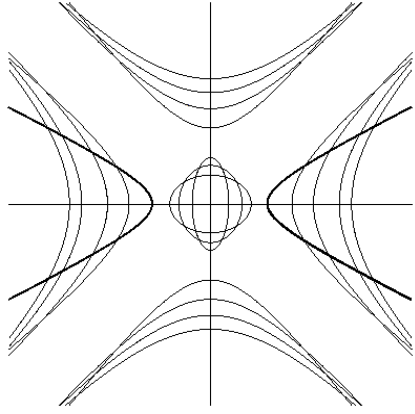}
\caption{Family of pseudo-confocal quadrics for $a_1>0,a_2<0$
in $\mathbb E^{1,1}$, where $\alpha_1=a_1>\alpha_2=-a_2$. }
\end{figure}

\begin{exm}{
If $A$ is positive definite, then
$\alpha_1>\dots>\alpha_g>0>\alpha_{g+1}>\dots>\alpha_r$ for some
index $g$. At a generic point $x\in \mathbb E^{k,l}$ we have $\langle
x,x\rangle_s>0$, $s=1,\dots,g$, $\langle x,x\rangle_s<0$,
$s=g+1,\dots,r$. Therefore, through a generic point $x\in \mathbb E^{k,l}$
pass either $r$ or $r-2$ quadrics from the pseudo--confocal
family \eqref{confocal} (see \cite{KT, DR}).
}\end{exm}

\begin{exm}{\label{MM}
Suppose that
\begin{equation}\label{MinMax}
\max\{a_1,\dots,a_k\}<\min\{-a_{k+1},\dots,-a_n\}.
\end{equation}
Then there is an index $g$, such that
\begin{equation}\label{NegPoz}
\mathbb E^{k_s,l_s}=\mathbb E^{0,l_s}, \quad s=1,\dots,g, \quad
\mathbb E^{k_s,l_s}=\mathbb E^{k_s,0}, \quad s=g+1,\dots,r,
\end{equation}
and through a generic point $x\in \mathbb E^{k,l}$ pass $r$ quadrics
from the confocal family \eqref{confocal} (Figure 2). On the other
hand, if
\begin{equation}\label{MaxMin}
\max\{-a_{k+1},\dots,-a_n\}<\min\{a_1,\dots,a_k\},
\end{equation}
then there is an index $g$, such that
\begin{equation}\label{PozNeg}
\mathbb E^{k_s,l_s}=\mathbb E^{k_s,0}, \quad s=1,\dots,g, \quad
\mathbb E^{k_s,l_s}=\mathbb E^{0,l_s}, \quad s=g+1,\dots,r,
\end{equation}
and through a generic point $x\in \mathbb E^{k,l}$ pass $r$ or $r-2$
quadrics.
}\end{exm}

\subsection{Geometrical interpretation of integrals.} The condition
\begin{equation}\label{jednacina}
\det \mathcal
L_{x,y}(\lambda)=q_\lambda(y,y)(1+q_\lambda(x,x))-q_\lambda(x,y)^2=0
\end{equation}
is equivalent to the geometrical property that the line
$$
l_{x,y}=\{x+sy \,\vert\, s\in\R\}
$$
is tangent to the quadric $\mathcal Q_{\lambda}$ (see \cite{Moser,
DR}).

Therefore, if the line $l_k$ determined by the segment $x_{k}x_{k+1}$
of the virtual billiard trajectory within $\mathbb Q^{n-1}$ is
tangent to a quadric $\mathcal Q_{\lambda^*}$, then $\det\mathcal
L_{x_k,y_k}(\lambda^*)=0$, implying $\det\mathcal
L_{x_k,y_k}(\lambda^*)=0$ for all $k$. Also note that
$\det\mathcal L_{x,y}(\lambda)$ is $SO(k_1,l_1)\times\dots\times
SO(k_r,l_r)$--invariant function.

As a result we have:

\begin{theorem}\label{ch1}\cite{JJ}\,
If a line $l_k$ determined by the segment $x_{k}x_{k+1}$ of the
virtual billiard trajectory within $\mathbb Q^{n-1}$ is tangent to
a quadric $\mathcal Q_{\lambda^*}$ from the pseudo--confocal
family \eqref{confocal}, then it is tangent to $\mathcal
Q_{\lambda^*}$ for all $k\in\mathbb Z$. In addition, $\mathbf
R(x_k)$ is a virtual billiard trajectory tangent to the same
quadric $\mathcal Q_{\lambda^*}$ for all $\mathbf R\in
SO(k_1,l_1)\times\dots\times SO(k_r,l_r)$.
\end{theorem}

From \eqref{SymL} follows that for a symmetric quadric \eqref{sym}  we have
\begin{align}
P(\lambda) &=
 (\lambda-\alpha_1)^{\delta_1} \cdots
(\lambda-\alpha_r)^{\delta_r}\det\mathcal L_{x,
y}(\lambda)  \label{symP2} \\
 &= \sum_{s=1}^r
\left((\lambda-\alpha_s)^{\delta_s-1}\prod_{i\ne s}
(\lambda-\alpha_i)^{\delta_i}F_s+\prod_{i\ne s}
(\lambda-\alpha_i)^{\delta_i}P_s\right) \nonumber \\
&= \lambda^{N-1}K_{N-1}+\dots+\lambda K_1+K_0,\nonumber
\end{align}
where
$$
\delta_s=2 \,\,\, \text{for} \,\,\, \vert I_s \vert \ge 2, \quad
\delta_s=1 \,\,\, \text{for} \,\,\, \vert I_s\vert=1, \quad
N=\delta_1+\dots+\delta_r.
$$

In particular, $K_{N-1}=2H=\langle y,y\rangle$. Thus, the degree
of $P(\lambda)$ is $N-1$ for a space--like or time--like vector
$y$, or $N-2$ for a light--like $y$, and for a general point
$(x,y)\in M_h$, the equation $\det\mathcal L_{x,y}(\lambda)=0$ has
either $N-1$ ($h\ne 0$) or $N-2$ ($h=0$) complex solutions. As in
the lemma above, the number of real solutions can be estimated in
certain cases. In \cite{JJ} we proved:

\begin{theorem}{\cite{JJ}\,}\label{posledica}
Suppose that $A$ is positive definite or the signature of the
space is $(n,0)$.  The lines determined by space--like or
time--like (respectively light--like) billiard trajectories
passing through generic points $(x,y)\in M_h$ are tangent to $N-1$
(respectively $N-2$) fixed quadrics from the pseudo--confocal
family \eqref{confocal}.
\end{theorem}

We proceed with the cases mentioned in the Example \ref{MM}.

\begin{theorem}\label{NOVA}
{\rm (i)} Suppose that the condition \eqref{MinMax} is satisfied. If
$EA$ is positive or negative definite, that is $\alpha_r>0$ or
$\alpha_1<0$, the lines determined by space--like or time--like
(respectively light--like) billiard trajectories passing through
generic points $(x,y)\in M_h$ are tangent to $N-1$ (respectively
$N-2$) fixed quadrics from the pseudo--confocal family
\eqref{confocal}.

{\rm (ii)} In the case when the condition \eqref{MaxMin} is satisfied
and $EA$ is positive (negative) definite, the lines determined by
generic time--like, light--like, space--like billiard trajectories
are tangent to at least $N-1$ ($N-3$), $N-2$, $N-3$
($N-1$) quadrics from the pseudo--confocal family
\eqref{confocal}, respectively.
\end{theorem}

\begin{proof} The proof is a modification of the idea used
in \cite{Au, DR} and \cite{JJ} for an analogous assertion in the
case of nonsymmetric ellipsoids and symmetric ellipsoids,
respectively. We have
\begin{equation}\label{jed1}
q_\lambda(y,y)=-\ds\sum_{i=1}^n\frac{y_i^2}{a_i-\tau_i\lambda}
=\ds-\sum_{s=1}^r\frac{\langle
y,y\rangle_s}{\alpha_s-\lambda}=\ds-\frac{R(\lambda)}{\ds\prod_{s=1}^r(\alpha_s-\lambda)},
\end{equation}
where
\[
R(\lambda)=\ds\sum_{s=1}^r \langle y,y\rangle_s\prod_{t\neq
s}(\alpha_t-\lambda)= (-1)^{r-1}\cdot\sum_{s=1}^r\langle
y,y\rangle_s\prod_{t\neq s}(\lambda-\alpha_t).
\]

From the definition of $R(\lambda)$ we obtain
\begin{equation}
\mbox{sign}\, R(\alpha_s)=\mbox{sign}\, \langle
y,y\rangle_s\,(-1)^{s+r}, \qquad s=1,\dots,r \label{znak3}
\end{equation}
and for a space--like or a time--like vector $y$:
\begin{align*}
&\mbox{sign}\, R(-\infty)=\mbox{sign}\langle y,y\rangle=\mbox{sign}\, h,\nonumber \\
&\mbox{sign}\, R(\infty)=(-1)^{r-1}\,\mbox{sign}\langle
y,y\rangle=(-1)^{r-1}\mbox{sign}\,h.
\end{align*}

Thus, for a space--like or a time--like vector $y$, we have
\begin{align}
&\mbox{sign}\, R(-\infty)\mbox{sign}\,
R(\alpha_r)=\mbox{sign}\,h\,\mbox{sign}\,  \langle y,y\rangle_r\nonumber\\
&\mbox{sign}\, R(\infty)\mbox{sign}\,
R(\alpha_1)=\mbox{sign}\,h\,\mbox{sign}\, \langle
y,y\rangle_1.\label{znak4}
\end{align}

Assume the relation $\alpha_r>0$. The proof for the case
$0>\alpha_1$ is the same.

\medskip

(i)  From \eqref{MinMax}, for a generic $(x,y)\in M_h$ we have
\begin{align}
&\mbox{sign} \langle y,y \rangle_1=\dots=\mbox{sign} \langle
y,y \rangle_g=-1, \nonumber \\
& \mbox{sign} \langle y,y \rangle_{g+1}=\dots=\mbox{sign} \langle
y,y \rangle_r=1,  \label{znak2}
\end{align}
for a certain index $g$.

From the relations  \eqref{znak3}, \eqref{znak4}, \eqref{znak2},
we obtain that the equation $R(\lambda)=0$ has $r-2$ solutions
$\zeta_s\in(\alpha_{s+1},\alpha_s)$ for
$s\in\{1,\dots,r-1\}\backslash\{g\}$ and another solution
$\zeta_r\in(-\infty,\alpha_r)$ (if $h<0$) or
$\zeta_0\in(\alpha_1,\infty)$ (if $h>0$).

Further, since $(x,y)\in M_h$, it follows that
$1+q_0(x,x)=1-(A^{-1}x,x)=0$, $q_0(x,y)=-(A^{-1}x,y)\ne 0$.
Whence,
$$
\det \mathcal L_{x,y}(0)=-q_0(x,y)^2<0.
$$

Thus, the left hand side of
\begin{equation}\label{qP}
\det \mathcal
L_{x,y}(\lambda)=q_\lambda(y,y)(1+q_\lambda(x,x))-q_\lambda(x,y)^2=\frac{P(\lambda)}{\prod_{s=1}^r(\lambda-\alpha_s)^{\delta_s}}
\end{equation}
takes negative values at the ends of each of the $r-2$ intervals
$$
(0,\zeta_{r-1}),(\zeta_{r-1},\zeta_{r-2}),\dots,(\zeta_{g+2},\zeta_{g+1}),
\,(\zeta_{g+1},\zeta_{g-1}),\,(\zeta_{g-1}, \zeta_{
g-2})\dots,(\zeta_2,\zeta_1),
$$
and
\begin{eqnarray*}
&& \alpha_r\in (0,\zeta_{r-1}),\, \alpha_{r-1}\in
(\zeta_{r-1},\zeta_{r-2}),\dots,\alpha_{g+2}\in(\zeta_{g+2},\zeta_{g+1}),
\\ && \alpha_{g+1},\alpha_g\in (\zeta_{g+1},\zeta_{g-1}),\,\alpha_{g-1}\in (\zeta_{g-1}, \zeta_{
g-2}),\dots,\alpha_2\in(\zeta_2,\zeta_1).
\end{eqnarray*}

From \eqref{MinMax}, we have that $\tau_i a_i=\tau_j a_j$ only if
$a_i=a_j$ and $\tau_i=\tau_j$. Hence, generically $P_s>0$ for
$\delta_s=2$. Now, from
\begin{equation}\label{bb}
\lim_{\lambda
\rightarrow\alpha_s-}\frac{F_s}{\lambda-\alpha_s}+\frac{P_s}{(\lambda
- \alpha_s)^2}=\infty,\quad \lim_{\lambda
\rightarrow\alpha_s+}\frac{F_s}{\lambda-\alpha_s}+\frac{P_s}{(\lambda
- \alpha_s)^2}=\infty,
\end{equation}
and \eqref{SymL}, it follows that in the interval containing
$\alpha_s$, $s\in\{2,3,\dots,r\}\backslash\{g,g+1\}$, there are at
least two zeros of $\det \mathcal L_{x,y}(\lambda)$ for
$\delta_s=2$ or at least one zero in the case $\delta_s=1$.
Similarly, in $(\zeta_{g+1},\zeta_{g-1})$ there are at least
$\delta_g+\delta_{g+1}-2$ zeros of $\det \mathcal
L_{x,y}(\lambda)$.

As a result, we get that in $(0,\zeta_1)$ there are
$$
\delta_2+\dots+\delta_{r-1}+\delta_r-2=N-\delta_1-2
$$
roots of $P(\lambda)$.

In the space--like case $h>0$, due to \eqref{znak4}, we have a root
$\zeta_0\in(\alpha_1,\infty)$ of $R(\lambda)$ and so there are
additional $\delta_1$ roots of $P(\lambda)$ in
$(\zeta_1,\zeta_0)$. Also, according to
\begin{equation}\label{beskonacno}
\det \mathcal
L_{x,y}(\lambda)=\frac{P(\lambda)}{\prod_{s=1}^r(\lambda-\alpha_s)^{\delta_s}}
\sim \frac{\langle y,y\rangle}{\lambda}, \qquad \lambda \to
\pm\infty,
\end{equation}
we have a zero of $\det\mathcal L_{x,y}(\lambda)$ in
$(\zeta_0,\infty)$ as well. Therefore, the number of real roots of
$P(\lambda)$ is $N-1$.

If $h<0$, thanks to \eqref{beskonacno}, there is a zero of
$\det\mathcal L_{x,y}(\lambda)$ in $(-\infty,0)$. Consequently, for $\delta_1=1$
we have at least $N-\delta_1-2 +1=N-2$ real roots of $P(\lambda)$. However,
since the polynomial $P(\lambda)$ is of degree $N-1$, it must have $N-1$ real roots.
By a similar argument there are $N-2$ real roots for $\delta_1=1$ and $h=0$.
If $\delta_1=2$ and $h<0$ or $h=0$, there is an additional zero of
$\det\mathcal L_{x,y}(\lambda)$ in $(\zeta_1,\alpha_1)$ and we can proceed as in the $\delta_1=1$ case.

\medskip

(ii) From \eqref{MaxMin}, for a generic $(x,y)\in M_h$ we have
\begin{eqnarray}
&&\mbox{sign} \langle y,y \rangle_1=\dots=\mbox{sign} \langle
y,y \rangle_g=1, \nonumber \\
&& \mbox{sign} \langle y,y \rangle_{g+1}=\dots=\mbox{sign} \langle
y,y \rangle_r=-1,  \label{znak2*}
\end{eqnarray}
for a certain index $g$. As above, we obtain that in $(0,\zeta_1)$
there are $\delta_2+\dots+\delta_{r-1}+\delta_r-2=N-\delta_1-2$
roots of $P(\lambda)$.

From \eqref{znak4}, \eqref{znak2*}, we have a root
$\zeta_0\in(\alpha_1,\infty)$ of $R(\lambda)$ for $h<0$. Hence
additional $\delta_1$ roots of $P(\lambda)$ in
$(\zeta_1,\zeta_0)$. Also, according to \eqref{beskonacno}, we
have a zero of $\det\mathcal L_{x,y}(\lambda)$ in $(-\infty,0)$
as well. Therefore, the number of real roots of $P(\lambda)$ is
$N-1$.

On the other hand, the analysis above in the space--like case
$h>0$ implies at least $N-3$ real roots of $P(\lambda)$. The
analysis for the light--like case $h=0$ is the same as in the
proof of (i).
\end{proof}

\begin{remark}{
In the previous proof we considered the case when $1<g<r$. The
borderline cases $g=1$ and $g=r$ have similar analysis. Moreover,
we have better estimates of the number of quadrics for the
assumptions \eqref{MaxMin} and $\delta_g=1$: if $EA$ is positive
(negative) definite and $g=1$ ($g=r$), then the signature of the
space is $(1,n-1)$ (respectively $(n-1,1)$) and there are $N-1$
caustics for billiard trajectories with $h\ne 0$ and $N-2$
caustics for $h=0$. This situation appears in Theorem \ref{posledica 2}.
}\end{remark}
\begin{exm}{
 Let us consider $\mathbb E^{1,1}$ and a nonsymmetric conic
defined by $A=\diag(a_1,a_2)$, $a_1>0>a_2$, $-a_2>a_1$ (see Figure
4). Then $\delta_1=\delta_2=1$, $\alpha_1=-a_2 >\alpha_2=a_1>0$
and from Lemma \ref{jacobi}, through the points $x=(x_1,x_2)$
outside the coordinate axes ($x_1\cdot x_2 \ne 0$) pass 2
quadrics from the family \eqref{confocal}.
In the non light--like case ($h=F_1+F_2\ne 0$), the
polynomial
$$
P(\lambda)=(\lambda-\alpha_1)(\lambda-\alpha_2)\left(\frac{F_1}{\lambda-\alpha_1}+\frac{F_2}{\lambda-\alpha_2}\right)=\lambda
h-\alpha_1\alpha_2 J^2
$$
has the real root $\lambda=(\alpha_1\alpha_2 J^2)/h$. This is a
root also in the case $\alpha_1=a_1>\alpha_2=-a_2>0$, as well (see
Figure 5).
}\end{exm}

\begin{figure}[ht]
\includegraphics[height=55mm]{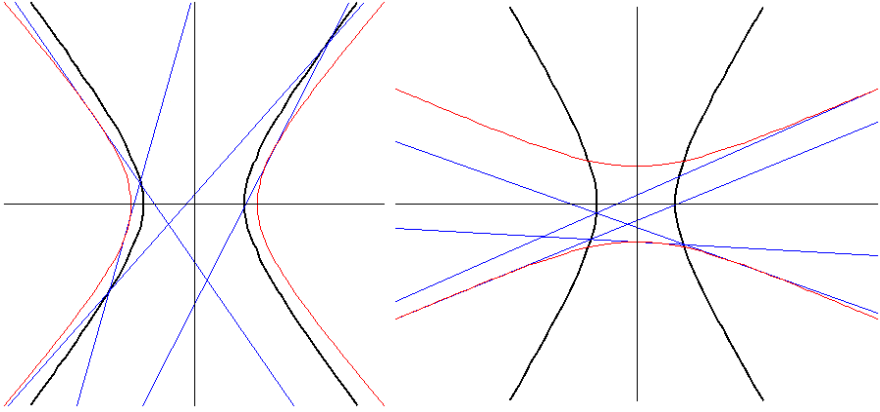}
\caption{The segments of time--like and space--like billiard
trajectories for $a_1>0,a_2<0$, $\alpha_1=-a_2>\alpha_2=a_1$ in
$\mathbb E^{1,1}$. The caustics are hyperbolas.}
\end{figure}

\begin{figure}[ht]
\includegraphics[height=28mm]{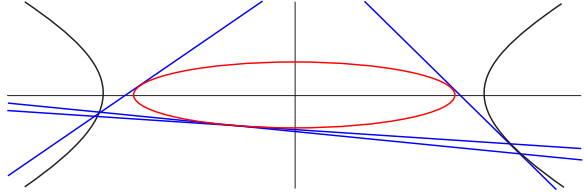} \caption{The segment of a space--like
billiard trajectory for $a_1>0,a_2<0$,
$\alpha_1=a_1>\alpha_2=-a_2$ in $\mathbb E^{1,1}$. The caustic is an
ellipse.}
\end{figure}

\begin{exm}{\rm
Next, we take $\mathbb E^{2,1}$ and a nonsymmetric quadric defined by
$A=\diag(a_1,a_2,a_3)$, $\alpha_1=-a_3>
\alpha_2=a_2>\alpha_3=a_1>0$. According to Lemma \ref{jacobi},
through the points $x=(x_1,x_2,x_3)$ outside the coordinate planes
($x_1\cdot x_2 \cdot x_3 \ne 0$) pass 3 quadrics from the
pseudo--confocal family \eqref{confocal}.
The discriminant of the polynomial
\begin{align*}
P(\lambda) &=
(\lambda-\alpha_1)(\lambda-\alpha_2)(\lambda-\alpha_3)
\left(\frac{F_1}{\lambda-\alpha_1}+\frac{F_2}{\lambda-\alpha_2}+\frac{F_3}{\lambda-\alpha_3}\right)\\
&= \lambda^2 h
-\lambda\left((\alpha_2+\alpha_3)F_1+(\alpha_3+\alpha_1)F_2+(\alpha_1+\alpha_2)F_3\right)
+\alpha_1\alpha_2\alpha_3 J^2
\end{align*}
equals
$D=((\alpha_2+\alpha_3)F_1+(\alpha_3+\alpha_1)F_2+(\alpha_1+\alpha_2)F_3)^2-4\alpha_1\alpha_2\alpha_3
h J^2. $ It is obvious that in the time--like case the
discriminant is positive and we always have two real roots. From
Theorem \ref{NOVA} (i) follows that $D>0$ in the
space--like case, too. In the light--like case, the real root
is
$$
\alpha_1\alpha_2\alpha_3
J^2/((\alpha_2+\alpha_3)F_1+(\alpha_3+\alpha_1)F_2+(\alpha_1+\alpha_2)F_3).
$$
}\end{exm}

Let us consider the signature $(n-1,1)$
in general situation. Suppose \eqref{alfici} and let
$g\in\{1,\dots,r\}$ be the index, such that $n\in I_g$. In order
to simplify the formulation of the theorem we additionally assume
that $\delta_g=1$, i.e., $I_g=\{n\}$.

\begin{theorem}\label{posledica 2}
Consider the lines determined by billiard trajectories in the
Lorentz--Poincar\'e--Minkowski space $\mathbb E^{n-1,1}$ passing through
generic points $(x,y)\in M_h$.
\begin{itemize}
\item[(i)] If $1<g<r$, the number of their caustics from the
pseudo--confocal family \eqref{confocal} is at least $N-3$
($h>0$), $N-1$ ($h<0$) or $N-4$ ($h=0$). If $\alpha_1<0$ and
$h>0$, the number of caustics is $N-1$.

\item[(ii)] Assuming $g=1$, there are $N-1$ quadrics ($h\neq 0$)
or at least $N-4$ quadrics ($h=0$). In addition, if $\alpha_r>0$
or $\alpha_1<0$, there are $N-2$ tangent quadrics for $h=0$.

\item[(iii)] In the case $g=r$ the minimal number of quadrics is
$N-3$, $N-1$ and $N-2$ for $h>0$, $h<0$ and $h=0$,
respectively. If we suppose $0\in(\alpha_r,\alpha_{r-1})$, then
there are $N-1$ quadrics for $h>0$, as well. If
$\alpha_r>0$ ($\alpha_1<0$) and $h>0$, $h<0$, $h=0$, the number of
caustics is at least $N-3$ ($N-1$), $N-1$ ($N-1$), $N-2$ ($N-2$), respectively.
\end{itemize}
\end{theorem}

\begin{proof} Let us prove the item (i). The proof of the
other statements is similar.

Since generically $\langle y,y\rangle_s>0$ for all $s\neq g$ and
$\langle y,y\rangle_g<0$, from \eqref{znak3} we have that there
exist $r-3$ solutions $\zeta_s\in(\alpha_{s+1},\alpha_s)$,
$s\in\{1,\dots,r-1\}\backslash\{g-1, g\}$ of the equation
$R(\lambda)=0$. Note that generically also $P_s>0$ for
$\delta_s=2$. Therefore, there are at least
$$
\delta_{r-1}+\dots
+\delta_{g+2}+\delta_{g-2}+\dots
+\delta_2=N-\delta_1-\delta_{g-1}-\delta_g-\delta_{g+1}-\delta_r
$$
zeros of $\det \mathcal L_{x,y}(\lambda)$ in the union
$(\zeta_{r-1},\zeta_{g+1}) \cup (\zeta_{g-2},\zeta_1)$. By
considering all the cases when $\delta_{g+1},
\delta_{g-1}\in\{1,2\}$, one concludes that the interval
$(\zeta_{g+1},\zeta_{g-2})$ contains at least
$\delta_{g+1}+\delta_g+\delta_{g-1}-2$ zeros, hence there are at
least $N-2-\delta_1-\delta_r$ zeros of  $\det \mathcal
L_{x,y}(\lambda)=0$ within the interval $(\zeta_{r-1},\zeta_1)$.

In the space-like case $h>0$, from \eqref{beskonacno} it follows
that there exists $\zeta_r<\alpha_r$,  such that $\det \mathcal
L_{x,y}(\zeta_r)<0$, whence additional $\delta_r$ zeros in
$(\zeta_r,\zeta_{r-1})$. On the other hand, in $(\zeta_1,\infty)$
lie at least $\delta_1-1$ zeros. In particular, if $\alpha_1<0$,
we have $\delta_1$ zeros in $(\zeta_1,0)$ and, thanks to
\eqref{beskonacno}, an additional zero in $(0,\infty)$.

If $h<0$, due to \eqref{znak4}, there are roots $\zeta_0>\alpha_1$
and $\zeta_r<\alpha_r$ of $R(\lambda)$ and, consequently,
$(\zeta_r,\zeta_0)$ has at least $N-2$ zeros of $\det \mathcal
L_{x,y}(\lambda)=0$. Further, from \eqref{beskonacno} it follows
that $(-\infty,\zeta_r)$ also has an additional zero of $\det
\mathcal L_{x,y}(\lambda)=0$.

Finally, for the light-like trajectories, by considering all the
cases when $\delta_r,\delta_1\in\{1,2\}$, the intervals
$(-\infty,\zeta_{r-1})$ and $(\zeta_1,\infty)$ have at least
$\delta_r-1$ and $\delta_1-1$ zeros, respectively.
\end{proof}

\begin{remark}{
Note that, if $g=r$, $\delta_g=1$ and
$0\in(\alpha_r,\alpha_{r-1})$, then $A$ is positive definite. On
the other hand, if $g=1$ and $0\in(\alpha_2,\alpha_1)$, then in
the case $\delta_1=1$, it is $\mathbb{Q}^{n-1}=\emptyset$, since
$a_i<0$ for all $i$.
}\end{remark}
\subsection{The Poncelet porism.}
Here, we suppose that one of the following conditions holds:

\begin{itemize}

\item[(i)] The signature is arbitrary, $A$ is positive definite.

\item[(ii)] The signature is $(n,0)$, $A$ is arbitrary.

\item[(iii)] The signature is arbitrary, $EA$ is positive or
negative definite and the assumption \eqref{MinMax} is satisfied.

\end{itemize}

Then $\tau_i a_i=\tau_j a_j$ only if $a_i=a_j$, $\tau_i=\tau_j$,
and the symmetry group is
\begin{equation}\label{grupa}
G=SO(\vert I_1\vert)\times \dots\times SO(\vert I_r\vert).
\end{equation}


From Theorems \ref{posledica}, \ref{NOVA} we get that, in the
space-like and the time--like cases, given a point $(x,y)\in M_h$
in a generic position, we have $N-1$ caustics
\begin{equation}\label{kaustike1}
\mathcal Q_{\lambda_1},\dots,\mathcal Q_{\lambda_{N-1}}
\end{equation}
determined by the
real zeros $\lambda_1,\dots,\lambda_{N-1}$ of $\det\mathcal
L_{x,y}(\lambda)$. They uniquely define the values of the
commuting integrals $F_s,P_s$ on $M_h$. Similarly, in a
light--like case, caustics
\begin{equation}\label{kaustike2}
\mathcal Q_{\lambda_1},\dots,\mathcal
Q_{\lambda_{N-2}}
\end{equation}
determined by the real zeros
$\lambda_1,\dots,\lambda_{N-2}$ of $\det\mathcal
L_{x,y}(\lambda)$, uniquely define the values of the commuting
integrals $\bar F_s,\bar P_s$ on $\bar M=M_0/\mathbb R$,  for a
generic $(x,[y])\in M_0$. Furthermore, all invariant isotropic
tori in $M_h$ with the same values of $F_s,P_s$, i.e., all
invariant pre-isotropic tori in $\bar M$ with the same values of
$\bar F_s,\bar P_s$, are related by the action of the groups of
symmetries $\mathbb Z^n_2$ (see \eqref{group}) and $G$ (see
\eqref{grupa}).

Therefore, by combining Theorems \ref{DIS}, \ref{integraliT},
\ref{posledica}, and \ref{NOVA}, we obtain:

\begin{theorem}\label{poncelet}
If a billiard trajectory $(x_k)$ is periodic with a period $m$ and
if the the lines $l_k$ determined by the segments $x_kx_{k+1}$ are
tangent to $N-1$ quadrics \eqref{kaustike1} (in the space--like or the time--like case) or
to $N-2$ quadrics \eqref{kaustike2} (in the light--like case), then any other
billiard trajectory within $\mathbb Q^{n-1}$ with the same
caustics is also periodic with the same period $m$.
\end{theorem}

Similarly, Theorem \ref{poncelet} applies also in all cases
described in Theorem \ref{NOVA} (ii) and Theorem
\ref{posledica 2} with maximal number of caustics.

\section{Pseudo--Euclidean billiards in projective spaces}

\subsection{Billiards on sphere and Lobachevsky space.}
It is well--known that the billiards within an ellipsoid $\mathbb E^{n-2}$ on the sphere $\mathbb S^{n-1}$ and the Lobachevsky space
$\mathbb H^{n-1}$ are completely integrable \cite{Bo, Ves2, Ta,
DGJ}. The ellipsoid $\mathbb E^{n-2}$ can be defined as a
intersection of a cone
\begin{equation}\label{cone}
\mathbb{K}^{n-1}: \qquad  (A^{-1}x,x)=0,
\end{equation}
where
\begin{equation}\label{EA}
A=\diag(a_1,\dots,a_n), \quad 0<a_1,a_2,\dots,a_{n-2},a_{n-1} <
-a_{n},
\end{equation}
with the Euclidean sphere
\begin{equation}\label{SPHERE}
\mathbb S^{n-1}=\{\langle x,x\rangle =1\}\subset \mathbb E^{n,0},
\end{equation}
or a connected component of a pseudosphere in the
Lorentz--Poincare--Minkowski space $\mathbb E^{n-1,1}$
\begin{equation}\label{lob}
\mathbb H^{n-1}=\{\langle x,x\rangle =-1, \,\, x_n>0\}\subset
\mathbb E^{n-1,1},
\end{equation}
respectively. The induced metrics on $\mathbb S^{n-1}$ and
$\mathbb H^{n-1}$ (a model of the Lobachevsky space) are
Riemannian with constant curvatures $+1$ and $-1$, while geodesic
lines are simply intersections of $\mathbb S^{n-1}$ and $\mathbb
H^{n-1}$ with two--dimensional planes through the origin.

Together with billiards on $\mathbb S^{n-1}$ and
$\mathbb H^{n-1}$, let us consider the following virtual billiard dynamic:
\begin{align}
x_{j+1}&=x_j+\mu_j y_j,\label{1bilijar*}\\
y_{j+1}&=y_j+\nu_j EA^{-1}x_{j+1}, \label{2bilijar*}
\end{align}
where the multipliers
$$
\mu_j=-2\,\frac{(\A x_j,y_j)}{(\A y_j,y_j)}, \qquad
\nu_j=2\,\frac{(\A
x_{j+1},y_{j+1})}{(EA^{-2} x_{j+1},x_{j+1})}
$$
are now determined from the conditions
$$
(\A x_{j+1},x_{j+1})=(\A x_j,x_j)=0, \qquad \langle y_{j+1},y_{j+1}\rangle=\langle y_j,y_j\rangle,
$$
that is, the impact points $x_j$ belong to the cone
\eqref{cone}. Again, the dynamics is defined outside the singular set
\begin{equation}\label{SINGULAR}
\Sigma=\{(x,y)\in T\R^n\,\,\vert\,\,(EA^{-2}x,x)=0 \,\, \vee
(A^{-1}x,y)=0\,\,\vee\,\,(A^{-1} y,y)=0\}.
\end{equation}

As a slight modification of Veselov's description of billiard
dynamics within $\mathbb E^{n-2}$ \cite{Ves2} we have the
following Lemma.

\begin{lemma}\label{opis}
Assume that the signature of the pseudo--Euclidean space $\mathbb E^{k,l}$
is $(n,0)$ or $(n-1,1)$, respectively.  Let $(x_j,y_j)$ be a
trajectory of the billiard mapping $\phi$ given by
\eqref{1bilijar*}, \eqref{2bilijar*}, where $A$ is given by
\eqref{EA}.
Then the intersections $z_j$ of the sequence of the lines
$\Span\{x_j\} $ with the ellipsoid $\mathbb E^{n-2}$ determine the
billiard trajectory within $\mathbb E^{n-2}$ on the sphere
$\mathbb S^{n-1}$ and the Lobachevsky space $\mathbb H^{n-1}$,
respectively.
\end{lemma}

\begin{proof}
Firstly,
we prove that the virtual billiard mapping $\phi$ defines the
dynamics of the lines $\Span\{x_j\}$, i.e, the dynamics of the
2-planes $\pi_j=\Span\{ x_j,y_j\}$ through the origin.

Consider
the transformation
\begin{equation}\label{promena}
x_j'=\alpha x_j, \quad y_j'=\beta x_j+\gamma y_j,\quad
\alpha,\beta,\gamma\in\R, \quad \alpha,\gamma\ne 0.
\end{equation}

Let $(x_{j+1},y_{j+1})$ and $(x_{j+1}',y_{j+1}')$ be respectively
the images of $(x_j,y_j)$ and $(x_j', y_j')$ with respect to the
mapping $\phi$. Since $x_j, y_j$ and $x_j', y_j'$ determine the
same 2-plane $\pi_j$, it follows that $x_{j+1}$ and $x_{j+1}'$ are
proportional and belong to $\pi_j \cap\mathbb K^{n-1}$. Thus, the
tangent planes $T_{x_{j+1}} \mathbb K^{n-1}$ and $T_{x_{j+1}'}
\mathbb K^{n-1}$ are equal and the corresponding billiard
reflections coincide.

Further, the incoming velocities $y_j$ and $y_j'$ also can be
related by $ y_j'=\beta' x_{j+1}+\gamma' y_j$, for certain
$\beta',\gamma'\in\mathbb R$. Since $x_{j+1}$ belongs to the
tangent plane $T_{x_{j+1}} \mathbb K^{n-1}$, after the reflections
$$
y_j \mapsto y_{j+1}, \qquad  y_j'\mapsto y_{j+1}'=\beta'
x_{j+1}+\gamma' y_{j+1},
$$
we get the unique 2-plane
\begin{equation*}\label{ravan}
\pi_{j+1}=\Span\{x_{j+1},y_{j+1}\}=\Span\{x_{j+1}',y_{j+1}'\}.
\end{equation*}

Concerning the singular set \eqref{SINGULAR},
note that the equation $(A^{-1}x,y)=0$
is invariant of the  mapping $\phi$ and, under condition \eqref{EA}, the only
solution of the equations
$(EA^{-2}x,x)=0$, $(A^{-1} x,x)=0$
is $x=0$. Also, if $(A^{-1} y_j,y_j)=0$, then we can apply
the transformation \eqref{promena} to obtain
$(A^{-1}y'_j,y'_j)=\beta\gamma (A^{-1}x_j,y_j)\ne 0$.

On the other hand, let $z_j,z_{j+1},z_{j+2}\in \mathbb \mathbb E^{n-2}$ be
3 successive points of the billiard trajectory within $\mathbb E^{n-2}$ and let $x_j=z_j, y_j=z_{j+1}-z_j$. Then
$$
\Span\{z_{j+1},z_{j+2}\}=\Span\{x_{j+1},y_{j+1}\},
$$
where $(x_{j+1},y_{j+1})=\phi(x_j,y_j)$, which completes the
proof.
\end{proof}

In \cite{DGJ}, Cayley's type conditions for periodical
trajectories of the ellipsoidal billiard on the Lobachevsky space
$\mathbb H^{n-1}$ are derived using the "big" $n\times n$-matrix
representation obtained by Veselov \cite{Ves2}. Here, as a simple
modification of the Lax representation \eqref{LAbil}, we obtain
the following "small" $2\times 2$--matrix representation of
billiards within $\mathbb E^{n-2}$. Note that the relationship
between the projective equivalence of the Euclidean space with the
Beltrami-Klein model of the Lobachevsky space and integrability of
the corresponding ellipsoidal billiards is obtained independently
in \cite{Ta} and \cite{DGJ}.

\begin{theorem}\label{billiardLA2}
The trajectories of the mapping \eqref{1bilijar*},
\eqref{2bilijar*} satisfy the matrix equation
\begin{equation}\label{billLA2}
\hat{\mathcal{L}}_{x_{j+1},y_{j+1}}(\lambda)=\hat{\mathcal{A}}_{x_j,y_j}(\lambda)\hat{\mathcal{L}}_{x_j,y_j}(\lambda)\hat{\mathcal{A}}_{x_j,y_j}^{-1}(\lambda),
\end{equation}
with $2\times2$ matrices depending on the parameter $\lambda$,
\begin{align*}
\hat{\mathcal{L}}_{x_j,y_j}(\lambda)&=\left(\begin{array}{cc}
q_{\lambda}(x_j,y_j) & q_{\lambda}(y_j,y_j)\\
-q_{\lambda}(x_j,x_j) & -q_{\lambda}(x_j,y_j)
\end{array}\right),\\
\hat{\mathcal{A}}_{x_j,y_j}(\lambda)&= \left(\begin{array}{cc}
I_j\lambda+2J_j\nu_j &-I_j\nu_j\\
-2J_j\lambda& I_j\lambda
\end{array}\right),
\end{align*}
where $q_{\lambda}$ is given by \eqref{bilin} and $J_j,I_j,\nu_j$
by \eqref{POMOC}.
\end{theorem}

\subsection{Billiards in projective spaces.}
Next, we consider the mapping \eqref{1bilijar*}, \eqref{2bilijar*} in the pseudo--Euclidean
spaces $\mathbb E^{k,l}$ of arbitrary signature and without the assumption
\eqref{EA}.
We also suppose the symmetries \eqref{sym}. Note that
Theorem \ref{billiardLA2} still applies and from the expression
\begin{equation}
\label{SymL2} \det\hat{\mathcal L}_{x,y}(\lambda)=\sum_{s=1}^r
\frac{\hat F_s}{\lambda-\alpha_s}+\frac{\hat P_s}{(\lambda - \alpha_s)^2},
\end{equation}
we get the integrals:
\begin{align*}
\label{integrals2} \hat F_s  &=\sum_{i\in I_s}\sum_{j\notin I_s}\frac{(x_iy_j-x_jy_i)^2}{\tau_j a_i-\tau_i a_j},\\
\nonumber          \hat P_s  &=\sum_{i,j\in I_s,i<j}\tau_i\tau_j(x_iy_j-x_jy_i)^2.
\end{align*}

They satisfy the relation
\begin{equation}\label{nula}
\hat F_1+\dots+\hat F_r=0.
\end{equation}

Further, as in the proof of Lemma \ref{opis}, if
$(x_j',y_j')$ is the image of $(x_j,y_j)$
by the transformation \eqref{promena} and
$(x_{j+1},y_{j+1})=\phi(x_j,y_j)$,
$(x_{j+1}',y_{j+1}')=\phi(x_j', y_j')$, then
the 2--planes spanned by $x_{j+1},y_{j+1}$ and $x_{j+1}',y_{j+1}'$ coincides.
Also, the part of
the singular set $\{(EA^{-2}x,x)=0\} \cup\{(A^{-1}x,y)=0\}$ in
\eqref{SINGULAR} is invariant with respect to the transformation
\eqref{promena}. If $(A^{-1} y_j,y_j)=0$, then we can apply
the transformation \eqref{promena} to obtain
$(A^{-1}y'_j,y'_j)=\beta\gamma (A^{-1}x_j,y_j)\ne 0$.
Thus, if necessary, we can
replace $y_j$ by $y_j'$ in order to determine $x_{j+1}$.

Therefore, the dynamics \eqref{1bilijar*}, \eqref{2bilijar*} induces
a well defined dynamics of the lines $\Span\{x_j\}$, i.e., the
points of the $(n-1)$--dimensional projective space $\mathbb
P(\mathbb E^{k,l})$
$$
z_j=[x_j]\in \mathbb Q^{n-2}
$$
outside the singular set
$$
\Xi=\{[x]\in \mathbb P(\mathbb E^{k,l})\,\,\vert\,\,
(EA^{-2}x,x)=0\},
$$
where $\mathbb Q^{n-2}$ is the projectivisation of the cone
\eqref{cone} within $\mathbb P(\mathbb E^{k,l})$.

\begin{definition} We refer to a sequence of the points $(z_j)$ as a
{\it billiard trajectory} within the quadric $\mathbb Q^{n-2}$ in
the projective space $\mathbb P(\mathbb E^{k,l})$ with respect to the
metric induced from the pseudo--Euclidean space $\mathbb E^{k,l}$.
\end{definition}

In particular, for signatures $(n,0)$ and $(n-1,1)$ with the
condition \eqref{EA} we obtain ellipsoidal billiards on the sphere
\eqref{SPHERE} and the Lobachevsky space \eqref{lob},
respectively.

Now we consider the following pseudo--confocal family of cones
(see \cite{Ves2})
\begin{equation}\label{confocal3}
\mathcal K_\lambda: \quad  ((A-E\lambda)^{-1} x,
x)=\sum_{i=1}^n\frac{x_i^2}{a_i-\tau_i\lambda}=0, \quad \lambda
\ne \tau_i a_i, \quad i=1,\dots,n,
\end{equation}
and the corresponding projectivisation, the pseudo--confocal
family of quadrics $\mathcal P_\lambda$.

\begin{theorem}\label{sfera}
Let $(z_k)$ be a sequence of the points of a billiard trajectory
within quadric $\mathbb Q^{n-2}$ in the projective space $\mathbb
P(\mathbb E^{k,l})$. If a projective line
$$
l_k=z_k z_{k+1}
$$
is tangent to a quadric $\mathcal P_{\lambda^*}$ then it is
tangent to $\mathcal P_{\lambda^*}$ for all $k\in\mathbb Z$.
\end{theorem}

\begin{proof} Let $\pi_I$, $I=(i_1,\dots, i_k)$, $1\le
i_1<i_2<\dots<i_k\le n$ be the Pl\"ucker coordinates of a
$k$-plane $\pi$ passing through the origin in ${\mathbb R}^{n}$.
Then  $\pi$ is tangent to the nondegenerate cone $\{\langle x,B
x\rangle =0\}$, $B={\rm diag}(b_1,\dots b_n)$ if and only if (see
Fedorov \cite{fe})
\begin{equation} \label{cone0}
\sum_I |B|_I^I \,\pi_I^2 =0, \qquad |B|_I^I =b_{i_1}\cdots
b_{i_k}.
\end{equation}

Now, let $z_k=[x_k]$,  $z_{k+1}=[x_{k+1}]$ and define
$y_k=x_{k+1}-x_k$, $ \pi_k=\Span\{x_k,y_k\}$. The condition that
the plane $\pi_k$ is tangent to the cone $\mathcal K_{\lambda^*}$
from the confocal family \eqref{confocal3} is given by the similar
invariant expression as in the case of virtual billiards within
quadric $\mathbb Q^{n-1}$,
\begin{equation}\label{jednacina2}
\det \hat{\mathcal L}_{x_k,y_k}(\lambda^*)=q_{\lambda^*}(y_k,y_k)q_{\lambda^*}(x_k,x_k)-q_{\lambda^*}(x_k,y_k)^2=0.
\end{equation}
Further, if $\det \hat{\mathcal L}_{x_k,y_k}(\lambda^*)=0$ for a given
$(x_k,y_k)$, it will be zero for all $k\in\mathbb Z$ under the
mapping $\phi$ (Theorem \ref{billiardLA2}), while from the
description of the billiard dynamics, the projectivisation of
$\pi_k=\Span\{x_k,y_k\}$ equals $l_k$ for all $k\in\mathbb Z$.

To obtain \eqref{jednacina2} we set
$$
B=\diag \left(\frac{1}{a_1-
\lambda^* \tau_1},\dots, \frac{1}{a_n- \lambda^* \tau_n}\right).
$$

Then, in view of \eqref{cone0}, the set of the $2$-planes
$\pi=\Span\{x,y\}$ that are tangent to $\mathcal K_{\lambda^*}$ is
described by the following quadratic equation in terms of the
Pl\"ucker coordinates $\pi_{i,j}=x_iy_j-x_jy_i$, $1\le i<j\le n$
of $\pi$
\begin{align*}
0 &=\sum_{1\le i<j\le n} \frac{ 1 }{ (a_i- \lambda^* \tau_i)
(a_j-
\lambda^* \tau_j ) } \, (x_iy_j-x_jy_i)^2  \\
&= \sum_{1\le i,j\le n} \frac{ 1 }{ (a_i- \lambda^*
\tau_i) (a_j- \lambda^* \tau_j ) } \, (x_i^2y_j^2-x_ix_jy_iy_j)\\
&=\sum \frac{ x_i^2 }{ (a_i- \lambda^* \tau_i)}\sum \frac{ y_i^2
}{ (a_i- \lambda^* \tau_i)}-\left(\sum \frac{ x_iy_i }{ (a_i-
\lambda^* \tau_i)}\right)^2=\det \hat{\mathcal L}_{x,y}(\lambda^*)\,.
\end{align*}
\end{proof}

In order to determine the number of caustics one should provide an
additional  analysis. The following situation leads to the
statement analogous to Theorems \ref{posledica} and \ref{NOVA}.

As in the case of the ellipsoidal billiards on a sphere $\mathbb
S^{n-1}$ and a Lobachevsky space $\mathbb H^{n-1}$, we assume the
relation \eqref{EA}. Then $\tau_i a_i=\tau_j a_j$ only if
$a_i=a_j$, $\tau_i=\tau_j$, $i,j<n$. As above, let $\delta_s=2$
for $\vert I_s \vert \ge 2$, $\delta_s=1$ for $\vert I_s\vert=1$,
and
$$
N=\delta_1+\dots+\delta_r.
$$

\begin{theorem}\label{ch-pr}
The lines $l_k=z_k z_{k+1}$ determined by a generic billiard
trajectory within $\mathbb Q^{n-2}$ are tangent to $N-2$ fixed
quadrics from the projectivisation of the confocal family
\eqref{confocal3}.
In particular, the trajectories of billiards within ellipsoid $\mathbb E^{n-2}$, with the above symmetry,
on the sphere
\eqref{SPHERE} and the Lobachevsky space \eqref{lob} are tangent to $N-2$ fixed
cones from the  confocal family
\eqref{confocal3}.
\end{theorem}

\begin{proof} From \eqref{SymL2},
\eqref{nula},
 we get
\begin{align*}
\hat P(\lambda)&=
 (\lambda-\alpha_1)^{\delta_1} \cdots
(\lambda-\alpha_r)^{\delta_r}\det\hat{\mathcal L}_{x,y}(\lambda)\\
&=\lambda^{N-2}\hat K_{N-2}+\dots+\lambda \hat K_1+\hat K_0.
\end{align*}

In addition, under the assumption \eqref{EA}, we can take
representatives $x_k, x_{k+1}$ of $z_k, z_{k+1}$, such that the
last components are equal to 1. Then, if we denote $x=x_k$ and
$y=x_{k+1}-x_k$, we have
$$
x=(x_1,\dots,x_{n-1},1), \qquad y=(y_1,\dots,y_{n-1},0).
$$

From \eqref{cone} we have $\det\hat{\mathcal L}_{x,y}(0)<0$ and
following the lines of the proof of Theorem \ref{NOVA}, it can be
proved that the equation $\hat P(\lambda)=0$ has
$N-2$ real solutions, for a generic $(x,y)$.
\end{proof}

Theorem \ref{ch-pr} for a nonsymmetric ellipsoid $\mathbb E^{n-2}$ ($N=n$) on the
Lobachevsky space $\mathbb H^{n-1}$ is well known (Theorem 3, \cite{Ves2}).

\subsection*{Acknowledgments}
We are grateful to the referee for pointing out several misprints.
The research of B. J. was supported by the Serbian Ministry of
Science Project 174020, Geometry and Topology of Manifolds,
Classical Mechanics and Integrable Dynamical Systems.

\end{document}